\documentclass[final,3p,times]{elsarticle}

\usepackage{amssymb}   
\usepackage{caption}   
\usepackage{subfigure} 
\usepackage{amsmath}   
\usepackage{setspace}  
\usepackage{algorithm} 
\usepackage{algpseudocode} 
\usepackage{makecell}  
\usepackage{multirow}  
\usepackage[justification=centering]{caption} 
\usepackage{multicol}
\usepackage{stfloats}  
\usepackage{color}     
\usepackage{amsthm}    
\usepackage{threeparttable}

\usepackage{dcolumn}

\usepackage{algorithm}
\usepackage{algpseudocode}

\newtheorem{definition}{Definition}
\newtheorem{example}{Example}
\newtheorem{theorem}{Theorem}
\newtheorem{lemma}{Lemma}

\begin{document}
\begin{frontmatter} 

%
\title{Nonoverlapping (delta, gamma)-approximate pattern matching}





	\author[1,2]{Youxi Wu}

		\author[1]{Bojing Jian}
		
		\author[3]{{Yan Li}\corref{mycorrespondingauthor}}
		\ead{lywuc@163.com}

		\author[4]{He Jiang}
		
		\author[5]{Xindong Wu}
		\cortext[mycorrespondingauthor]{Corresponding author}
		
\address[1]{School of Artificial Intelligence, Hebei University of Technology, Tianjin 300401, China}

\address[2]{Hebei Key Laboratory of Big Data Computing, Tianjin {\rm 300401}, China}

\address[3]{School of Economics and Management, Hebei University of Technology, Tianjin 300401, China}
\address[4]{School of Software, Dalian University of Technology, Dalian, Liaoning 116621, China}
\address[5]{Key Laboratory of Knowledge Engineering with Big Data ( Ministry of Education), Hefei University of Technology, Hefei 230009, China}


\begin{abstract}
Pattern matching can be used to calculate the support of patterns, and is a key issue in sequential pattern mining (or sequence pattern mining). Nonoverlapping pattern matching means that two occurrences cannot use the same character in the sequence at the same position. Approximate pattern matching allows for some data noise, and is more general than exact pattern matching. At present, nonoverlapping approximate pattern matching is based on Hamming distance, which cannot be used to measure the local approximation between the subsequence and pattern, resulting in large deviations in matching results. To tackle this issue, we present a Nonoverlapping Delta and gamma approximate Pattern matching (NDP) scheme that employs the $(\delta, \gamma)$-distance to give an approximate pattern matching, where the local and the global distances do not exceed $\delta$ and $\gamma$, respectively. We first transform the NDP problem into a local approximate Nettree and then construct an efficient algorithm, called the local approximate Nettree for NDP (NetNDP). We propose a new approach called the Minimal Root Distance which allows us to determine whether or not a node has root paths that satisfy the global constraint and to prune invalid nodes and parent-child relationships. NetNDP finds the rightmost absolute leaf of the max root, searches for the rightmost occurrence from the rightmost absolute leaf, and deletes this occurrence. We iterate the above steps until there are no new occurrences. Numerous experiments are used to verify the performance of the proposed algorithm.
\end{abstract}


\end{frontmatter}

\section{Introduction}
	Pattern matching (or string matching \cite{Al-Ssulami2015Hybrid,Fernau2020Pattern, qiangtkde, wu2015apin}) is a key issue in computer science \cite{ida2020,Nie2016Query, li2021dse,Wang2019} and plays an important part in many applications, such as pattern mining \cite {Wu2020NetNCSP,Min2018Frequent,song2021kais}, knowledge discovery \cite{Wu2019On}, bioinformatics analysis \cite{Upama2015A}, fault detection \cite{Lee2018Fault}, and time series forecasting \cite{Nguyen2017Pattern}. In recent years, pattern matching with multiple gap constraints, as an important branch of pattern matching \cite {wu2021jos,ida2018}, has attracted a great deal of attention in many fields, such as time series analysis \cite{wu2021tkdd,min2020ins}, sequential pattern mining \cite{apin2021,Fournier2021}, and text key phrase extraction \cite{Xie2017Efficient}. A pattern with gap constraints can be written as $P=p_{1}[min_{1}, max_{1}]p_{2} \cdots p_{j}[min_{j}, max_{j}]p_{j + 1}$ $\cdots p_{m-1}[min_{m-1}, max_{m-1}]p_{m}$, where $min_{j}$ and $max_{j}$ represent the minimal and maximal wildcards between $p_{j}$ and $p_{j+1}$, respectively. $[min_{j}, max_{j}]$ can be set flexibly \cite{Liu2021}, and this appoach is more practical than the traditional wildcards ``?'' and ``*''.

	In a pattern matching problem with gap constraints, the number of occurrences is exponential if there are no constraints \cite{Huang2013Algorithms}, and some researchers have therefore focused on pattern matching under the one-off condition \cite{Wu2021eswa}. However, the one-off condition is too stringent, and means that some useful information is lost. Ding et al. \cite{Ding2009Efficient} then proposed the concept of nonoverlapping. The nonoverlapping pattern matching was proved to be solved in polynomial time complexity \cite{Wu2017Strict}. To make the gap constraint more suitable for practical applications, non-overlapping pattern matching with general gaps was studied \cite{Shi2020NetNPG}. Nonoverlapping sequential pattern mining is able to find valuable frequent patterns effectively \cite{Wu2018NOSEP}. To reduce the number of redundant patterns, nonoverlapping closed sequential pattern matching was explored and improved the mining performance \cite{Wu2020NetNCSP}.
	
	However, the above mentioned researches are exact pattern matching \cite{Wu2017Strict,Chen2014Bit}. It means that noise is not allowed in the above researches, which is difficult to obtain valuable information. Approximate pattern matching \cite{Hu2015GFilter,Chen2018On} allows for noise, and it can therefore handle practical problems more effectively.  The Hamming distance \cite{Wu2018NETASPNO,Wu2016Approximate} is commonly applied in approximate pattern matching. But the Hamming distance only reflects the number of different characters, and ignores the distances between characters. Therefore, nonoverlapping approximate pattern matching scheme based on the Hamming distance may cause large deviations when applied to time series  \cite{li2021apind,wu2020APIN}.  Inspired by $(\delta, \gamma)$-distance, \cite{Zhang2017On,Fredriksson2006Efficient,Clifford2005},  this paper focuses on Nonoverlapping Delta and gamma approximate Pattern matching (NDP) that employs the $(\delta, \gamma)$-distance to give an approximate pattern matching, where the local and the global distances do not exceed $\delta$ and $\gamma$, respectively. An illustrative example is shown as follows.

Figure \ref {9figs} shows the matching results of symbolized time series of pattern $P=$ b$[0,1]$d$[0,1]$b. Figure 1(a) is consistent with pattern $P$ without gaps, while Figures 1(b)(c) contain gaps and can match pattern $P$ exactly. Figures 1(d)(e)(f) match pattern $P$ with a Hamming distance of one, but due to the large deviations, they are not very similar to Figure 1(a). For instance, `e' and `b' show a large deviation in Figure 1(d). Figures 1(g)(h)(i) which match pattern $P$ with the $(\delta, \gamma)$-distance show a close similarity to Figure 1(a). Figures 1(g)(h) match pattern $P$ with $(\delta=1, \gamma=1)$, while Figure 1(i) matches pattern $P$ with $(\delta=1, \gamma=2)$. We can therefore see that the $\delta$-distance and $\gamma$-distance can be used to measure the local approximation and the global approximation, respectively. $(\delta, \gamma)$-approximate pattern matching ensures overall similarity.

	\begin{figure} 
		\centering
		\includegraphics[width=0.75\textwidth]{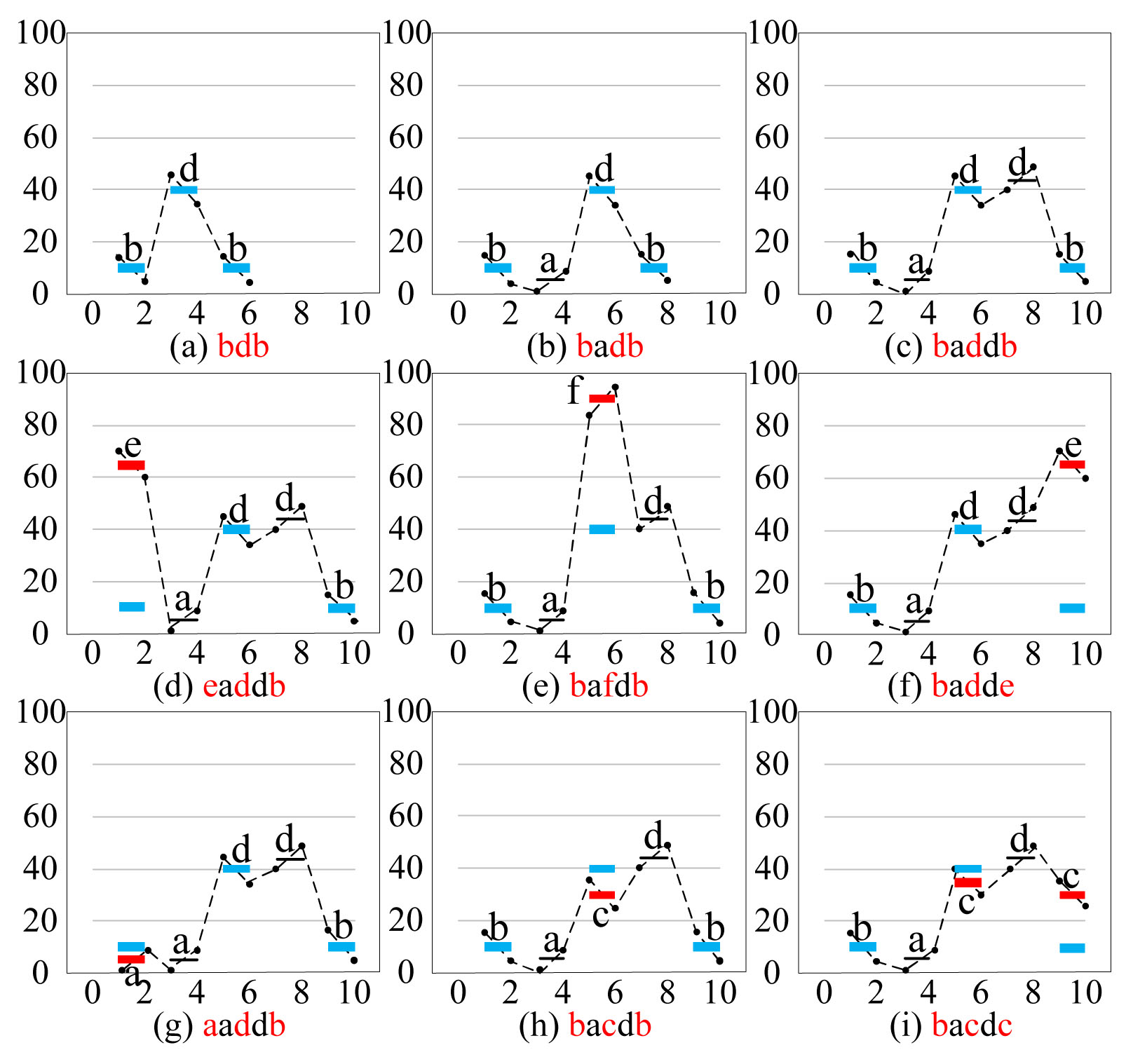}	
	
		Note: ``\includegraphics[width=2.5mm,height=0.35mm]{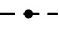}'' indicates a numerical time series, ``\protect\includegraphics[width=2.5mm,height=1mm]{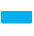}'' indicates exact matching segments, ``\protect\includegraphics[width=2.5mm,height=1mm]{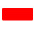}'' indicates approximate matching segments, ``\protect\includegraphics[width=2.5mm,height=0.6mm]{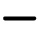}'' indicates gaps.
		
		\caption{{Pattern $P=$ b$[0,1]$d$[0,1]$b and symbolized time series matching results. The time series in (a), (b), (c), (d), (e), (f), (g), (h), and (i) can be symbolized as ``bdb'', ``badb'', ``baddb'', ``ebaddb'', ``bafdb'', ``badde'', ``aaddb'', ``bacdb'', and , ``bacdc'', respectively.   All these sequences can be matched by pattern $P$. The trends of (b) and (c) are similar to that of (a), since they are exact pattern matching. Although Figures 1(d)(e)(f) match pattern $P$ with a Hamming distance of one, the trends of (d), (e), and (f) are significant different from that of (a) due to the large deviations. However, Figures 1(g)(h)(i) match pattern P with the $(\delta, \gamma)$-distance, and the trends of (g), (h), and (i) are similar to that of (a).}
		}
		\label{9figs} 	
	\end{figure}
	
		The contributions of this paper are as follows.
	
	(1) To avoid  large deviations of Hamming distance, we present a Nonoverlapping Delta and gamma approximate Pattern matching (NDP) shceme and propose  an efficient algorithm called NetNDP (a local approximate Nettree for NDP).
	
	(2)  NetNDP employs the concept of a local approximate Nettree and MRD (Minimal Root Distance) to improve the efficiency.
	
	(3) We carry out numerous experiments to verify the efficiency of NetNDP, and show that the $(\delta, \gamma)$-distance outperforms the Hamming distance for matching.  
	
	The rest of this paper is organized as follows:  Section 2 presents some relevant definitions. Section 3 introduces related work. Section 4 explores the local approximate Nettree structure, and proposes the NetNDP algorithm. Section 5 presents results that validate the performance of NetNDP. We conclude this paper in Section 6.

	\section{Problem Definition}
	\label{problem_definition}
	\begin{definition}
		A sequence $S$ can be written as $S=s_{1}s_{2} \cdots s_{n}$, where $n$ is the length of $S$, $s_{i} \in \Sigma (1 \leq i \leq n)$, and $\Sigma$ is the character set.
	\end{definition}
	
	\begin{definition}
		A pattern with gap constraints can be expressed as $P=p_{1}[min_{1}$, $max_{1}]p_{2} \cdots p_{j}[min_{j},$ $max_{j}]p_{j + 1} \cdots p_{m-1}[min_{m-1}, max_{m-1}]p_{m}$, where $p_{j} \in \Sigma$ $(1 \leq j \leq m)$, $m$ is the length of $P$, and $min_{j}$ and $max_{j}$ represent the minimal and maximal wildcards between $p_{j}$ and $p_{j+1}$, respectively. $[min_{j}, max_{j}]$ is a gap constraint.
	\end{definition} 
	
	\begin{definition}
		Given any two characters $c$ and $d$ in $\Sigma$ , the $\delta$-distance between $c$ and $d$ is $|c-d|$, and is denoted by $D_{\delta}(c, d)$.
	\end{definition} 

	\begin{definition}
		Suppose we have two sequences $S_{1}=x_{1}x_{2} \cdots x_{n}$ and $S_{2}=y_{1}y_{2} \cdots y_{n}$, where  $x_{i} \in \Sigma$,  $y_{i} \in \Sigma (1 \leq i \leq n)$. The $\gamma$-distance between $S_{1}$ and $S_{2}$ is $\sum_{i=1}^{n} D_{\delta}(x_{i}, y_{i})=\sum_{i=1}^{n}|x_{i}-y_{i}|$, and is denoted by $D_{\gamma}(S_{1}, S_{2})$.
	\end{definition} 
	
\begin{example}
	 Suppose $S_{1}=$ aef and $S_{2}=$ cee, the $\delta$-distances between the corresponding characters of $S_{1}$ and $S_{2}$ are then $D_{\delta}(x_{1}, y_{1})=|\mathrm{a}-\mathrm{c}|=2$, $D_{\delta}(x_{2}, y_{2})=|\mathrm{e}-\mathrm{e}|=0$, and $D_{\delta}(x_{3}, y_{3})=|\mathrm{f}-\mathrm{e}|=1$. Thus, the $\gamma$-distance between $S_{1}$ and $S_{2}$ is $D_{\gamma}(S_{1}, S_{2})=\sum_{i=1}^{3} D_{\delta}(x_{i}, y_{i})=\sum_{i=1}^{3}|x_{i}-y_{i}|=2+0+1=3$. 
\end{example}
	
	\begin{definition}
		Suppose we have a sequence $S=s_{1}s_{2} \cdots s_{n}$, a pattern $P=p_{1}[min_{1}$, $max_{1}]p_{2} \cdots [min_{j},$ $max_{j}]p_{j+1} \cdots [min_{m-1}, max_{m-1}]p_{m}$, a local threshold $\delta$, and a global threshold $\gamma$. If $<l_{1}, l_{2}, \cdots, l_{m}>$  satisfies the following conditions,  then $<l_{1}, l_{2}, \cdots, l_{m}>$ is a $(\delta, \gamma)$-approximate occurrence of $P$ in $S$.

		(1) $1 \leq l_{1}<l_{2}<\cdots<l_{m}\leq n$ and $min_{j} \leq l_{j+1}-l_{j}-1 \leq max_{j}$  $(1 \leq j \leq m-1)$.
		
		(2) (local constraint) $\max_{j=1}^{m} D_{\delta}(s_{l_{j}}, p_{j}) \leq \delta$.
		
		(3) (global constraint) $D_{\gamma}(s_{l_{1}}s_{l_{2}} \cdots s_{l_{m}}, p_{1}p_{2} \cdots p_{m})=\sum_{j=1}^{m}|s_{l_{j}}-p_{j}| \leq \gamma$.
		
	\end{definition}
	
	\begin{definition}
		$L=$ $<l_{1}, l_{2}, \cdots, l_{m}>$ and $L^{\prime}=$ $<l_{1}^{\prime}, l_{2}^{\prime}, \cdots, l_{m}^{\prime}>$ are any two occurrences of pattern $P$ in sequence $S$. If for all $j$ $(1 \leq j \leq m)$ $l \neq l_{j}^{\prime}$, then $L$ and $L^{\prime}$ are two nonoverlapping occurrences.
	\end{definition}
	
	\begin{definition}
		The NDP problem is to find the maximum $(\delta
		, \gamma)$-approximate occurrences of pattern $P$ in sequence $S$, where any two occurrences are nonoverlapping.
	\end{definition} 
	
\begin{example}
 We have a sequence $S=s_{1}s_{2}s_{3}s_{4}s_{5}=$ acaba, a pattern $P=p_{1}[min_{1}$, $max_{1}]p_{2}[min_{2}, max_{2}]$ $p_{3}=$ a$[0,1]$b$[0,2]$a, a local threshold $\delta=1$, and a global threshold $\gamma=1$. According to Figure \ref{occurrences}, we know that there are four $(\delta, \gamma)$-approximate occurrences of $P$ in $S$ under no special condition, which are $<$1, 2, 3$>$, $<$1, 3, 5$>$, $<$1, 4, 5$>$, and $<$3, 4, 5$>$. The maximum nonoverlapping $(\delta, \gamma)$-approximate occurrences are  $\{<$1, 2, 3$>$, $<$3, 4, 5$>\}$. 
\end{example}



	\section{Related Work}	\label{related_work}


Pattern matching can be divided into exact \cite{Wu2017Strict,Chen2014Bit} and approximate pattern matching \cite{Hu2015GFilter}. Exact pattern matching requires that the pattern and subsequence are the same, which limits the flexibility of matching, while approximate pattern matching allows differences between the pattern and subsequence, and can discover more useful information. For instance, Hu et al. \cite{Hu2015GFilter} implemented a string similarity search through a gram filter. Chen et al. \cite{Chen2018On} discovered all subsequences with $k$ mismatches using an indexing mechanism. The Hamming distance \cite{Wu2018NETASPNO,Wu2016Approximate} is commonly applied in approximate pattern matching. However, it only reflects the number of different characters, and ignores the distances between characters. Motivated by this, Zhang et al. \cite{Zhang2017On} developed threshold-based approximate pattern matching, in which the maximal distance between characters does not exceed $\delta$. Fredriksson and Grabowski \cite{Clifford2005} added a threshold $\gamma$, meaning that the sum of the distances between the characters does not exceed $\gamma$, which can obtain a better matching effect. Compared with traditional wildcard pattern matching, the pattern matching with gap constraints is both more flexible and difficult.  An illustrative example is shown as follows.
	
\begin{example}
 Suppose we have a sequence $S=s_{1}s_{2}s_{3}s_{4}s_{5}=$ acaba, a pattern $P=p_{1}[min_{1}$, $max_{1}]p_{2}[min_{2}, max_{2}]p_{3}=$ a$[0,1]$b$[0,2]$a, a local threshold $\delta=1$ and global threshold $\gamma=1$.	Figure \ref{occurrences} shows all $(\delta, \gamma)$-approximate occurrences of pattern $P$ in sequence $S$. 

	\begin{figure} 
		\centering
		\includegraphics[width=0.75\textwidth]{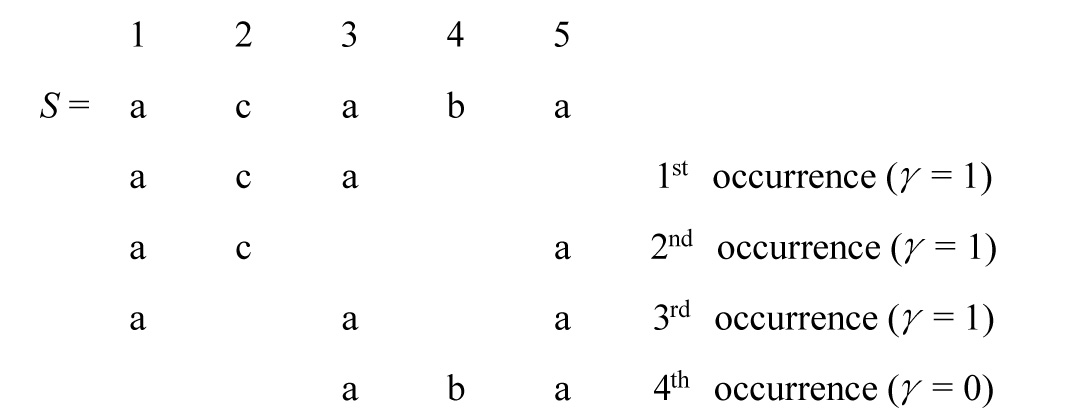}
		\caption{All $(\delta, \gamma)$-approximate occurrences of  $P = a[0,1]b[0,2]a$ in  $S$ with $(\delta=1, \gamma=1)$
		}
		\label{occurrences} 	
	\end{figure}

The subsequence $s_{1}s_{2}s_{3}$ is abbreviated as $<$1, 2, 3$>$, and the other three occurrences are $<$1, 2, 5$>$, $<$1, 3, 5$>$, and $<$3, 4, 5$>$. The reasons are as follows. Let us consider $<$3, 4, 5$>$ at first. Since $s_{3}=p_{1}=$ a, we know that $D_{\delta}(s_{3}, p_{1})=0$. Similarly, $D_{\delta}(s_{4}, p_{2})=D_{\delta}(s_{5}, p_{3})=0$, Therefore, $D_{\delta}(s_{3}s_{4}s_{5}, p_{1}p_{2}p_{3})=0$.  Hence, $<$3, 4, 5$>$ is an exact match which is a special case of $(\delta, \gamma)$-approximate pattern matching.
Now, let us take $<$1, 2, 3$>$ as another example. Since we know that $D_{\delta}(s_{2}, p_{2})=|\mathrm{c}-\mathrm{b}|=1 \leq \delta$, and $D_{\delta}(s_{1}, p_{1})=D_{\delta}(s_{3}, p_{3})=0$, we also know that $D_{\gamma}(s_{1}s_{2}s_{3}, p_{1}p_{2}p_{3})=1 \leq \gamma$. Thus, $<$1, 2, 3$>$ is a legal $(\delta, \gamma)$-approximate occurrence with $(\delta=1, \gamma=1)$. Similarly, $<$1, 2, 5$>$ and $<$1, 3, 5$>$ are two legal $(\delta, \gamma)$-approximate occurrences with $(\delta=1, \gamma=1)$. However, $<$2, 4, 5$>$ and $<$1, 3, 4$>$ are illegal $(\delta, \gamma)$-approximate occurrences with $(\delta=1, \gamma=1)$, since $D_{\delta}(s_{2}, p_{1})=2>\delta$ and $D_{\gamma}(s_{1}s_{3}s_{4}, p_{1}p_{2}p_{3})=2>\gamma$, respectively.
\end{example}

	From the perspective of description an occurrence, there are loose  pattern matching and strict pattern matching. Loose pattern matching uses the last position in the sequence to express an occurrence \cite{Clifford2005}, while strict pattern matching uses a set of positions to express. Under loose pattern matching \cite{Clifford2005}, there are only two $(\delta, \gamma)$-approximate occurrences in Example 1, $<$3$>$ and $<$5$>$, while Example 1 can be seen as strict pattern matching. Obviously, compared with loose pattern matching, strict pattern matching emphasizes the details of matching.
	
	There are three kinds of conditions in strict pattern matching, no special condition \cite{Wu2016Approximate}, the one-off condition \cite{Wu2021eswa}, and the nonoverlapping condition \cite{Wu2018NETASPNO}. Example 3 can be seen as  no special condition, since the condition has no restrictions on characters, which result in the number of occurrences grows exponentially. The occurrence under the one-off condition may be any of the four occurrences, since any character of the sequence can only be used once. For instance, if we select $<$1, 2, 3$>$, then $<$3, 4, 5$>$ cannot be selected since this means that $s_{3}$ would be used twice. However, $<$1, 2, 3$>$ and $<$3, 4, 5$>$ are two occurrences under the nonoverlapping condition. Although $s_{3}$ is used twice, it matches with $p_{3}$ and $p_{1}$, respectively. We can clearly see that the nonoverlapping condition can simplify the occurrences without excluding useful information, and therefore reduces the limitations of matching.

{Although nonoverlapping pattern matching has been investigated in \cite{Wu2017Strict}, it is an exact pattern matching approach which does not allow noise. To the best of our knowledge, the researches in \cite{Clifford2005,Wu2018NETASPNO,wu2020APIN} are the most closest studies. The drawbacks are as follows.}

{(1) Although the scheme in \cite{Clifford2005} focused on approximate pattern matching with the $(\delta, \gamma)$-distance, it is a kind of loose pattern matching which adopts the end position in the sequence to represent an occurrence. Therefore, it is difficult to represent an occurrence clearly. Our method is a kind of strict pattern matching which uses a set of positions in the sequence to express an occurrence. Thus, our method can represent an occurrence clearly. Hence, our method is a more practical approach.}

{(2) Although nonoverlapping approximate pattern matching has been investigated in \cite{Wu2018NETASPNO}, it adopts the Hamming distance to measure the distance between pattern and occurrence. But the Hamming distance only reflects the number of different characters, and cannot evaluate the distances between characters. However, our method employs $(\delta, \gamma)$-distance to effectively measure the local and global distances. Therefore, our method is more challenging. }

{(3) Although $(\delta, \gamma)$-distance pattern mining under no special condition was studied in \cite {wu2020APIN}, this method can find all approximate occurrences, which contain many redundant occurrences. As Example 3 shown, there are four $(\delta, \gamma)$-approximate occurrences according to no special condition \cite {wu2020APIN}. However, Example 3 has two nonoverlapping $(\delta, \gamma)$-approximate occurrences, which can effectively reduce redundant occurrences.}

{Our method can be used in repetitive sequential pattern mining \cite{dong2018tcyb, wu2021tmis} and sequence classification \cite{Wu2021tcyb,wu2019cc}. For example, the repetitive sequential pattern mining task is to discover frequent patterns in sequences whose supports are no less than the given threshold \cite{ Truong2020EHAUSM,Fournier2020Mining}. To calculate the support of a candidate pattern is a pattern matching task. Thus, many gap constraints sequential pattern mining methods employed pattern matching strategies to discover the interesting patterns \cite{li2021apin,wu2021kbs}. Moreover, contrast pattern mining was investigated to extract the features for classification task \cite{Wu2021tcyb,wu2019cc}. Hence, our method can also be applied in approximate repetitive sequential pattern mining and time series classification. }


	\section{Local Approximate Nettree and Algorithms}
	\label{algorithms}
	Section 4.1 introduces the local approximate Nettree. Section 4.2 explains how to transform the NDP problem into a local approximate Nettree and proposes the NetNDP algorithm. 
	
	\subsection{Local Approximate Nettree}
	\label{local_approximate_nettree}
	\begin{definition}
		A Nettree \cite{Wu2016Approximate} is an extended tree structure with multiple roots and multiple parents. In a Nettree, nodes with the same ID can be found at different levels. To describe a node clearly, a node with ID $i$ at the $j^{\rm th}$ level is denoted by $n_{j}^{i}$.
	\end{definition} 
	
{A standard Nettree is shown in Figure \ref{anettree}. In Figure 3, $r_1,\cdots, r_m$ are the roots of the Nettree. $T_1, T_2, \cdots, T_n$ are the subNettrees. Each subNettree can have many parents.}
 \begin{figure} 
	 	\centering
	 	\includegraphics[width=0.25\textwidth]{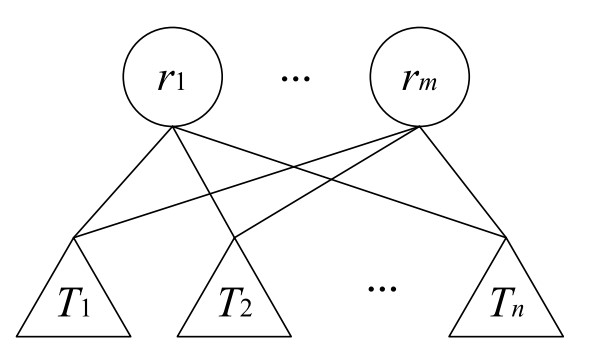}	\\ 	
	 	
	 	\caption{{A standard Nettree}}
	 	\label{anettree} 	
	 \end{figure}


	\begin{definition}
		In a Nettree, a leaf at the $m^{\rm th}$ level is called an absolute leaf.
	\end{definition} 
	
	\begin{definition}
		In a Nettree, a path from $n_{j}^{i}$ to a root is called a root path of $n_{j}^{i}$. A path from a absolute leaf to a root is called a root-leaf path, or a full path, and is denoted by $<n_{1}^{i_{1}}, n_{2}^{i_{2}}, \cdots, n_{m}^{i_{m}}>$.
	\end{definition} 
	
	\begin{lemma}
		An occurrence of pattern $P$ in sequence $S$ can be represented by a full path.
	\end{lemma} 
\begin{proof}
 We know that   an occurrence can be transformed into a full path in a Nettree \cite {Wu2016Approximate}, i.e. a full path $<n_{1}^{i_{1}}, n_{2}^{i_{2}}, \ldots, n_{m}^{i_{m}}>$ corresponds to an occurrence $<i_{1}, i_{2}, \cdots, i_{m}>$. Thus, a search for the occurrences of pattern $P$ in sequence $S$ can be transformed into a search for the full paths in a Nettree.
\end{proof}
	
	\begin{definition}
		In a Nettree, $n_{j}^{i}$ can reach multiple parents, of which the max parent is the rightmost parent of $n_{j}^{i}$. In a similar way, we can reach the rightmost absolute leaf.
	\end{definition} 
	
	To effectively solve the problem of nonoverlapping approximate pattern matching with the $(\delta, \gamma)$-distance, we propose the concept of a local approximate Nettree.
	
	\begin{definition} \label {lan}
		Local approximate Nettree has two features  different from a Nettree.	
	
		(1) Each node $n_{j}^{i}$ calculates and stores its $D_{\delta}(s_{i}, p_{j})(D_{\delta}(s_{i}, p_{j}) \leq \delta)$.
		
		(2) Each node $n_{j}^{i}$ calculates and stores its MRD, denoted by $M_{r}(n_{j}^{i})$, which presents the shortest $\gamma$-distance from $n_{j}^{i}$ to its roots.
	\end{definition} 
	
	\begin{theorem}
		If $n_{j-1}^{r^{1}}, n_{j-1}^{r^{2}}, \cdots,$ and $n_{j-1}^{r^{t}}$ are parents that satisfy the gap constraint  $[min_{j-1}, max_{j-1}]$ with $n_{j}^{i}$, then $M_{r}(n_{j}^{i})$ can be calculated as follows.
		\begin{small} 
			\begin{equation*}
			M_{r}(n_{j}^{i}) =\left\{
			\begin{array}{crr}
			D_{\delta}(s_{i}, p_{1})=|s_{i}-p_{1}|,\quad j=1\\
			\min (M_{r}(n_{j-1}^{r^{1}}), M_{r}(n_{j-1}^{r^{2}}), \cdots, M_{r}(n_{j-1}^{r^{t}}))+D_{\delta}(s_{i}, p_{j}),\quad 2 \leq j \leq m 
			\end{array}
			\right.
			\end{equation*}
		\end{small}
		where $t$ denotes the number of parents of $n_{j}^{i}$.
	\end{theorem}
\begin{proof}
 According to Definition \ref{lan}, $M_{r}(n_{j}^{i})$ reprents the shortest $\gamma$-distance from $n_{j}^{i}$ to its roots. If $j=1$, then $n_{j}^{i}$ is a root, and $M_{r}(n_{j}^{i})$ is the $\gamma$-distance from $n_{1}^{i}$ to itself, i.e. $D_{\delta}(s_{i}, p_{1})=|s_{i}-p_{1}|$. If $2 \leq j \leq m$, then the $\gamma$-distance from $n_{j}^{i}$ to its roots is the sum of $D_{\delta}(s_{i}, p_{j})$ and the $\gamma$-distance from its parents to its roots. Since $D_{\delta}(s_{i}, p_{j})$ is a fixed value, we need to find the shortest $\gamma$-distance from the parents to the roots. We can see that the $\gamma$-distance from the parent with the minimal MRD to the roots is the shortest, and the sum of the $\gamma$-distance and $D_{\delta}(s_{i}, p_{j})$ is $M_{r}(n_{j}^{i})$.
\end{proof}

To illustrate the above concepts, an example is  shown as follows.

\begin{example} \label{exmp4}
Suppose we have a sequence $S=s_{1}s_{2}s_{3}s_{4}s_{5}s_{6}s_{7}s_{8}s_{9}=$ baabcbbab, a pattern $P=p_{1}[min_{1}, max_{1}]p_{2} [min_{2},$ $ max_{2}]p_{3}[min_{3}, max_{3}]p_{4}=$ b$[0,1]$a$[0,2]$b $[0,2]$b, a local threshold $\delta$=1, and a global threshold $\gamma$=1. The corresponding local approximate Nettree is shown in Figure \ref{nettree1}.
	 
	 \begin{figure} 
	 	\centering
	 	\includegraphics[width=0.35\textwidth]{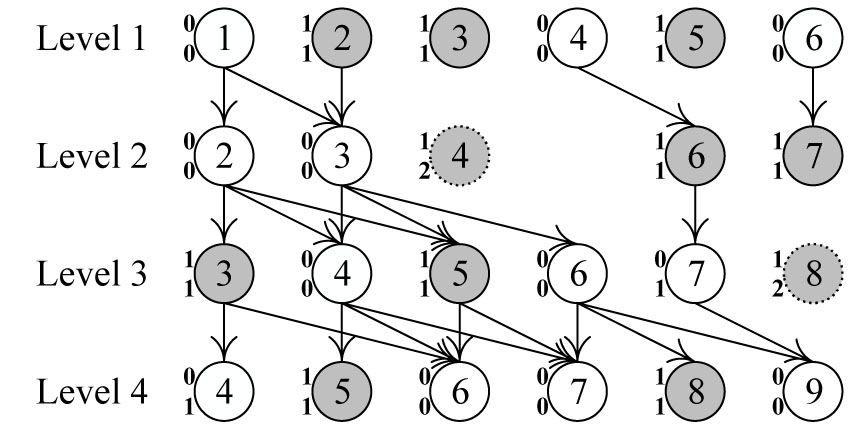}	\\ 	
	 	Note: The white node $n_{j}^{i}$ indicates that $s_{i}$ matches $p_{j}$ exactly. The grey node $n_{j}^{i}$ indicates that $s_{i}$ matches $p_{j}$ approximately. There are two numbers outside each node: the top left and bottom left numbers which are $D_{\delta}(s_{i}, p_{j})$ and $M_{r}(n_{j}^{i})$, respectively.
	 	
	 	\caption{A local approximate Nettree}
	 	\label{nettree1} 	
	 \end{figure}
	
	In Figure  \ref{nettree1}, nodes with the same ID occur at different levels, such as $n_{1}^{3}$, $n_{2}^{3}$, and $n_{3}^{3}$. $<n_{1}^{4}, n_{2}^{6}, n_{3}^{7}, n_{4}^{9}>$ is a full path, and its corresponding occurrence is $<$4, 6, 7, 9$>$. $s_{3}$ matches $p_{2}$ exactly, and thus $n_{2}^{3}$ is a white node. $s_{3}$ matches $p_{3}$ approximately, and thus $n_{3}^{3}$ is a grey node. Each node $n_{j}^{i}$ in the local approximate Nettree satisfies $D_{\delta}(s_{i}, p_{j}) \leq 1$, and each node $n_{1}^{i}$ satisfies $M_{\delta}(n_{1}^{i})=D_{\delta}(s_{i}, p_{1})$. Since $M_{r}(n_{4}^{9})=0$, the shortest $\gamma$-distance of $n_{4}^{9}$ to its roots is zero, and the corresponding path is therefore $<n_{1}^{1}, n_{2}^{3}, n_{3}^{6}, n_{4}^{9}>$. For path $<n_{1}^{2}, n_{2}^{3}, n_{3}^{6}, n_{4}^{8}>$, $n_{1}^{2}$ is the rightmost parent of $n_{2}^{3}$.
\end{example}

	\subsection{NetNDP}	\label{sec:NetNDP}
\subsubsection{Creating the Nettree}\label{sec:the_creating_algorithm}
	There are two key issues in creating a local approximate Nettree: creating nodes and parent-child relationships. We propose Theorem 2 to reduce the invalid nodes, and explore Theorem 3 to reduce parent-child relationships.

	\begin{theorem}
		If $M_{r}(n_{j}^{i})>\gamma$, then $n_{j}^{i}$  can be deleted.
	\end{theorem} 
\begin{proof}
We know that $D_{\gamma}(s_{l_{1}}s_{l_{2}} \cdots s_{l_{j-1}}s_{i}, p_{1}p_{2}\cdots p_{j-1}p_{j}) \geq D_{\gamma}(s_{l_{1}}s_{l_{2}} \cdots s_{l_{j-1}},$ $ p_{1}p_{2}\cdots p_{j-1})$. Thus, the $\gamma$-distance is monotonic. Suppose $M_{r}(n_{j}^{i})>\gamma$, i.e. the shortest $\gamma$-distance from $n_{j}^{i}$ to its roots is greater than $\gamma$. Since the $\gamma$-distance is monotonic, the $\gamma$-distance of all root-leaf paths via $n_{j}^{i}$ is greater than $\gamma$. Therefore, there is no path via $n_{j}^{i}$ that satisfies the global constraint. Hence, if $M_{r}(n_{j}^{i})>\gamma$, $n_{j}^{i}$ can be deleted.
\end{proof}
	
	\begin{theorem}
		If $n_{j-1}^{r^{q}}$ satisfies the gap constraint $[min_{j-1}, max_{j-1}]$ with $n_{j}^{i}$, and $M_{r}(n_{j-1}^{r^{q}})+D_{\delta}(s_{i}, p_{j}) > \gamma$, then a parent-child relationship cannot be established between $n_{j-1}^{r^q}$ and $n_{j}^{i}$, otherwise the parent-child relationship is created.
	\end{theorem} 
\begin{proof}
 Suppose $M_{r}(n_{j-1}^{r^{q}})+D_{\delta}(s_{i}, p_{j}) > \gamma$, i.e. the sum of $D_{\delta}(s_{i}, p_{j})$ and the shortest $\gamma$-distance from $n_{j-1}^{r^{q}}$ to its roots is greater than $\gamma$, meaning that the sum of $D_{\delta}(s_{i}, p_{j})$ and the $\gamma$-distance from $n_{j-1}^{r^{q}}$ to any root is greater than $\gamma$. Since the $\gamma$-distance is monotonic, the $\gamma$-distance for all root-leaf paths via $n_{j-1}^{r^{q}}$ and $n_{j}^{i}$ is greater than $\gamma$. Thus, there is no path via $n_{j-1}^{r^{q}}$ and $n_{j}^{i}$ that satisfies the global constraint.  Hence, the parent-child relationship between node $n_{j-1}^{r^{q}}$ and node $n_{j}^{i}$ cannot be established. 
\end{proof}

\begin{example}
 To illustrate the principles of Theorems  1, 2, and 3, we apply the same conditions as in Example \ref {exmp4}.
	
	We first deal with $s_{1}$. Since $D_{\delta}(s_{1}, p_{1})=|\mathrm{b}-\mathrm{b}|=0 \leq \delta$, we create $n_{1}^{1}$ at the first level. According to Theorem 1, we know that $M_{r}(n_{1}^{1})=D_{\delta}(s_{1}, p_{1})=0$. 

We then turn to $s_{2}$. Since $D_{\delta}(s_{2}, p_{1})=|\mathrm{a}-\mathrm{b}|=1 \leq \delta$, we create $n_{1}^{2}$ at the first level, and $M_{r}(n_{1}^{2})=D_{\delta}(s_{2}, p_{1})=1$. Since $D_{\delta}(s_{2}, p_{2})=|\mathrm{a}-\mathrm{a}|=0 \leq \delta$, we create $n_{2}^{2}$ at the second level. There is a parent $n_{1}^{1}$ that satisfies the gap constraint [0, 1] with $n_{2}^{2}$. From Theorem 1, we know that $M_{r}(n_{2}^{2})=M_{r}(n_{1}^{1})+D_{\delta}(s_{2}, p_{2})=0 \leq \gamma$. From Theorem 3, we establish a parent-child relationship between $n_{1}^{1}$ and $n_{2}^{2}$.

We now deal with $s_{3}$. Since $D_{\delta}(s_{3}, p_{1})=|\mathrm{a}-\mathrm{b}|=1 \leq \delta$, we create $n_{1}^{3}$ at the first level, and $M_{r}(n_{1}^{3})=D_{\delta}(s_{3}, p_{1})=1$. Similarly, we create $n_{2}^{3}$ at the second level, since $D_{\delta}(s_{3}, p_{2})=|\mathrm{a}-\mathrm{a}|=0 \leq \delta$. Both $n_{1}^{1}$ and $n_{1}^{2}$ satisfy the gap constraint [0, 1] with $n_{2}^{3}$. Thus, according to Theorem 1, $M_{r}(n_{2}^{3})=\min(M_{r}(n_{1}^{1}), M_{r}(n_{1}^{2}))+D_{\delta}(s_{3}, p_{2})=0 \leq \gamma$, and according to Theorem 3, $n_{2}^{3}$ establishes a parent-child relationship with both $n_{1}^{1}$ and $n_{1}^{2}$. 

We create the rest of the nodes in a similar way. Since $M_{r}(n_{2}^{4})=\min(M_{r}(n_{1}^{2}),$ $M_{r}(n_{1}^{3}))+D_{\delta}(s_{4}, p_{2})=2>\gamma$ and $M_{r}(n_{3}^{8})=\min(M_{r}(n_{2}^{6}), M_{r}(n_{2}^{7}))+D_{\delta}(s_{8}, p_{3})=2>\gamma$, we know from Theorem 2 that $n_{2}^{4}$ and $n_{3}^{8}$ can be deleted. Using this method, the local approximate Nettree can be created.
\end{example}
	
	From Example 5, we see that MRD has the following three advantages.

 (i) It allows us to know whether or not a node has root paths that satisfy the global constraint. If $M_{r}(n_{j}^{i})$ is not greater than $\gamma$, then $n_{j}^{i}$ has root paths that satisfy the global constraint. For instance, we can see that $M_{r}(n_{4}^{9})=0 \leq \gamma$, and that a root path that satisfies the global constraint on $n_{4}^{9}$ is $<n_{1}^{1}, n_{2}^{3}, n_{3}^{6}, n_{4}^{9}>$. 

(ii) Some invalid nodes can be pruned, such as $n_{2}^{4}$. 

(iii) We can prune invalid parent-child relationships. For instance, a parent-child relationship cannot be established between $n_{3}^{3}$ and $n_{4}^{5}$.

Algorithm 1, called CreLANtree, is used to create the local approximate Nettree for a sequence $S$, a pattern $P$, a local threshold $\delta$, and a global threshold $\gamma$.

	\begin{algorithm}[t] 
		\caption{CreLANtree} 
		\hspace*{0.02in} {\bf Input:} 
		sequence $S$, pattern $P$, local threshold $\delta$, global threshold $\gamma$\\
		\hspace*{0.02in} {\bf Output:} 
		$LANtree$
		\begin{algorithmic}[1]
			\For{$i = 1$ to $n$ step 1}
			\If{$D_{\delta}(s_{i}, p_{1}) \leq \delta$}
			\State create $n_{1}^{i}$ and $M_{\delta}(n_{1}^{i}) \leftarrow   D_{\delta}(s_{j}, p_{1})$;
			\EndIf
			\For{$j = 2$ to $m$ step 1}
			\If{$D_{\delta}(s_{i}, p_{j}) \leq \delta$}  
			\State Create $n_{j}^{i}$;
			\State Update $M_{r}(n_{j}^{i})$ according to Theorem 1;
			\If{$M_{r}(n_{j}^{i})>\gamma$} 
			\State Delete $n_{j}^{i}$ according to Theorem 2;
			\Else
			\State Establish parent-child relationships between its parents and  $n_{j}^{i}$ according to Theorem 3;
			\EndIf
			\EndIf
			\EndFor
			\EndFor
			\State return $LANtree$;
		\end{algorithmic}
	\end{algorithm}

	\subsubsection{Searching for nonoverlapping $(\delta, \gamma)$-approximate occurrences}
	\label{sec:Searching_occurrences}
	Section 4.2.1 explains the principle used to create a local approximate Nettree, where the corresponding occurrences of all full paths satisfy the local constraint. In this section, we will introduce the principle of searching for the nonoverlapping occurrences that satisfy the global constraint.
	
	\begin{lemma}
		Let $L_{1}$ and $L_{2}$ be two root-leaf paths that do not involve the same node, the corresponding occurrences of $L_{1}$ and $L_{2}$ are then nonoverlapping. 
	\end{lemma}
\begin{proof}
Suppose we have $L_{1}=$ $<n_{1}^{a_{1}}, n_{2}^{a_{2}}, \cdots, n_{m}^{a_{m}}>$ and $L_{2}=$ $<n_{1}^{b_{1}}, n_{2}^{b_{2}}, \cdots,$ $n_{m}^{b_{m}}>$. The corresponding occurrences of $L_{1}$ and $L_{2}$ are $<a_{1}, a_{2}, \cdots, a_{m}>$ and $<b_{1}, b_{2}, \cdots, b_{m}>$, respectively. $L_{1}$ and $L_{2}$ do not involve the same node, i.e. for any $j$ $(1 \leq j \leq m), a_{j} \neq b_{j}$. Thus, according to Definition 6, $<a_{1}, a_{2}, \cdots, a_{m}>$ and $<b_{1}, b_{2}, \cdots, b_{m}>$ are two nonoverlapping occurrences.
\end{proof}
	
	In a local approximate tree, when we search for occurrences from an absolute leaf to its roots, we first assess whether or not the rightmost parent satisfies the condition. If not, we will assess the second rightmost parent, until a qualified parent is found. This is known as the rightmost parent strategy.
	
	For nonoverlapping exact pattern matching, NETLAP-Best \cite{Wu2017Strict} adopts the rightmost parent strategy, and iteratively searches for the rightmost occurrence of the max absolute leaf, and deletes both the occurrence and the related invalid nodes. However, this method cannot be employed to solve our problem, since this method is too blind, which will lead to the loss of solution. An illustrative example is shown as follows.

\begin{example}
 We use the same conditions as in Example 4.

In Figure  \ref{nettree1}, suppose the global threshold $\gamma=1$. We first find occurrence $<$4, 6, 7, 9$>$ with a $\gamma$-distance of one, which satisfies the global constraint. We delete $<$4, 6, 7, 9$>$, and then find occurrence $<$2, 3, 6, 8$>$ from $n_{4}^{8}$. However, the $\gamma$-distance of $<$2, 3, 6, 8$>$ is two, which is greater than $\gamma$. We therefore deselect the rightmost parent $n_{1}^{2}$ of $n_{2}^{3}$, and instead select the second rightmost parent $n_{1}^{1}$ to obtain $<$1, 3, 6, 8$>$ with a $\gamma$-distance of one. We delete $<$1, 3, 6, 8$>$, and no other occurrences can be found after that. Using NETLAP-best, we only obtain two nonoverlapping $(\delta, \gamma)$-approximate occurrences. However, in Figure  \ref{nettree1}, there are three nonoverlapping $(\delta, \gamma)$-approximate occurrences: $<$1, 2, 5, 6$>$, $<$2, 3, 6, 7$>$, and $<$4, 6, 7, 9$>$. Hence, the principle of NETLAP-Best cannot be applied to NDP.
\end {example}

To avoid the drawbacks of  NETLAP-Best algorithm, NetNDP applies two steps to obtain the rightmost occurrence. In the first step, we recalculate the MRD of each node in the subNettree of the max root, and judge whether the root can reach the absolute leaves under the condition of the global constraint. If so, we obtain the rightmost absolute leaf of the max root. In the second step we get the rightmost occurrence using the rightmost parent strategy with the rightmost absolute leaf. We iterate the above process until no new occurrences are found.

	The ReachLeaf algorithm, which obtains the rightmost absolute leaf from the max root, is shown in Algorithm 2.
	\begin{algorithm}[t] 
		\caption{ReachLeaf} 
		\hspace*{0.02in} {\bf Input:} 
		$LANtree$, root $R$, local threshold $\delta$, global threshold $\gamma$\\
		\hspace*{0.02in} {\bf Output:} 
		the rightmost absolute leaf $ral$ of $R$
		\begin{algorithmic}[1]
			\State $lal \leftarrow ral \leftarrow R$;  //  $lal$ and $ral$ are used to indicate the range from the leftmost node to the rightmost node traversed by each layer
			\For{$i = 1$ to $m-1$ step 1}
			\For{$j = lal$ to $ral$ step 1}
			\For{$k = 1$ to $LANtree[i][j]$.children.size() step 1}
			\State Recalculate the MRD of $LANtree[i+1][k]$ in the subnettree of $R$ according to Theorem 1; 
			\EndFor
			\EndFor
			\State $lal \leftarrow $first child at the ${i+1}^{\rm th}$ level, $ral \leftarrow$ last child at the ${i+1}^{\rm th}$ level;
			\If{$lal ==$ NULL}
			\State return -1;
			\EndIf
			\EndFor
			\State return $ral$;       
			
		\end{algorithmic}
	\end{algorithm}
	
	After obtaining the rightmost absolute leaf, we prove that we can obtain the rightmost occurrence without the need for a backtracking strategy.
	
	\begin{theorem}
		In the subNettree of a root, if the root can reach $n_{m}^{i}$ under the condition of the global constraint, we can obtain a full path from $n_{m}^{i}$ to the root without a backtracking strategy.
	\end{theorem} 

\begin{proof}
If the root can reach $n_{m}^{i}$ under the condition of the global constraint, we know that $M_{r}(n_{m}^{i}) \leq \gamma$. Otherwise, according to Theorem 2, $n_{m}^{i}$ needs to be deleted and if $M_{r}(n_{m}^{i})>\gamma$, the shortest $\gamma$-distance from $n_{m}^{i}$ to the root is greater than $\gamma$, i.e. there is no path from $n_{m}^{i}$ to the root satisfies the global constraint. Thus, $M_{r}(n_{m}^{i}) \leq \gamma$. Suppose $D_{\delta}(s_{i}, p_{m})=k \leq \gamma$. We then need to search for a path from the parents of $n_{m}^{i}$ to the root with a $\gamma$-distance within $\gamma-k$. From Theorem 1, we know that $M_{r}(n_{m}^{i})= \min(M_{r}(n_{m-1}^{r^{1}}), M_{r}(n_{m-1}^{r^{2}}),$  $\cdots, M_{r}(n_{m-1}^{r^{t}}))+k \leq \gamma$. Hence, from the rightmost parent to the leftmost parent of $n_{m}^{i}$, there must be a parent whose MRD is no greater than $\gamma-k$. We select the parent, and iterate the process until the first level is reached. In other words, when searching for a path from $n_{m}^{i}$ to the root with a $\gamma$-distance of within $d$, we seek the rightmost parent whose MRD is not greater than $d$. Since the first level has only one root in the subNettree, we can obtain a full path from $n_{m}^{i}$ to the root without the need for a backtracking strategy.
\end{proof}
	
	Based on Theorem 4, we develop the RightOcc algorithm to obtain the rightmost occurrence, as shown in Algorithm 3.

	\begin{algorithm}[t] 
		\caption{RightOcc} 
		\hspace*{0.02in} {\bf Input:} 
		$LANtree$, rightmost absolute leaf $ral$, local threshold $\delta$, global threshold $\gamma$\\
		\hspace*{0.02in} {\bf Output:} 
		a nonoverlapping $(\delta, \gamma)$-approximate occurrence $occ$
		\begin{algorithmic}[1]
			\State $occ[m] \leftarrow LANtree[m][ral]$;
			\For{$i = m-1$ to 1 step -1}
			\State $occ[i] \leftarrow $ the rightmost parent of current under the condition of the global constraint; 
			\EndFor
			\State return $occ$;	
		\end{algorithmic}
	\end{algorithm}
	
	Example 7 illustrates the principle of NetNDP.
	
\begin{example}
 We use the same conditions as in Example 4.

	In Figure \ref{nettree3}, $n_{1}^{6}$ is the max root. It cannot reach the absolute leaves, and neither can $n_{1}^{5}$. Next, we assess root $n_{1}^{4}$, and recalculate the MRD for each node in its subnettree. From Figure  \ref{nettree3}, we can clearly see that the rightmost absolute leaf of $n_{1}^{4}$ is $n_{4}^{9}$, and $M_{r}(n_{4}^{9})=0 \leq \gamma$. Hence, according to Theorem 4, we find the $(\delta, \gamma)$-approximate occurrence $<$4, 6, 7, 9$>$ from $n_{4}^{9}$.

\begin{figure}[t]
		\centering
		\includegraphics[width=0.35\textwidth]{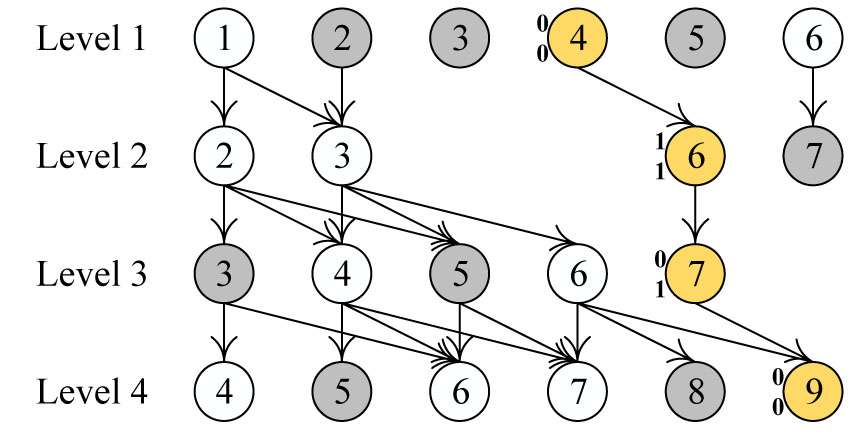}
		\caption{MRD for each node in the subNettree of $n_{1}^{4}$}
		\label{nettree3} 	
	\end{figure}

	We delete $<$4, 6, 7, 9$>$ and mark it in red. Now,  the rightmost absolute leaf of $n_{1}^{2}$ is $n_{4}^{7}$. Hence, we start searching for an occurrence from $n_{4}^{7}$ using the rightmost parent strategy. When searching for a path from $n_{4}^{7}$ to $n_{1}^{2}$ with a $\gamma$-distance of within one, we seek the rightmost parent whose MRD is no greater than one. Since $D_{\delta}(s_{7}, p_{4})=0$, $n_{3}^{6}$ is the rightmost parent, and $M_{r}(n_{3}^{6})=1$, meaning that there exists a root path with a $\gamma$-distance of one from $n_{3}^{6}$ to $n_{1}^{2}$. We therefore select $n_{3}^{6}$. This process is iterated until we finally find the $(\delta, \gamma)$-approximate occurrence $<$2, 3, 6, 7$>$ according to Theorem 4.

	We now delete $<$2, 3, 6, 7$>$ and assess the last root $n_{1}^{1}$, as shown in Figure  \ref{nettree4}. We recalculate the MRD for each node in the subNettree of $n_{1}^{1}$. We can see that $M_{r}(n_{4}^{6})=0 \leq \gamma$, and the rightmost absolute leaf of $n_{1}^{1}$ is $n_{4}^{6}$ after deleting $<$2, 3, 6, 7$>$. According to Theorem 4, we find the $(\delta, \gamma)$-approximate occurrence $<$1, 2, 5, 6$>$ from $n_{4}^{6}$. In summary, there are three nonoverlapping $(\delta, \gamma)$-approximate occurrences of pattern $P$ in sequence $S$: $<$1, 2, 5, 6$>$, $<$2, 3, 6, 7$>$, and $<$4, 6, 7, 9$>$.

	\begin{figure}[t]  
		\centering
		\includegraphics[width=0.35\textwidth]{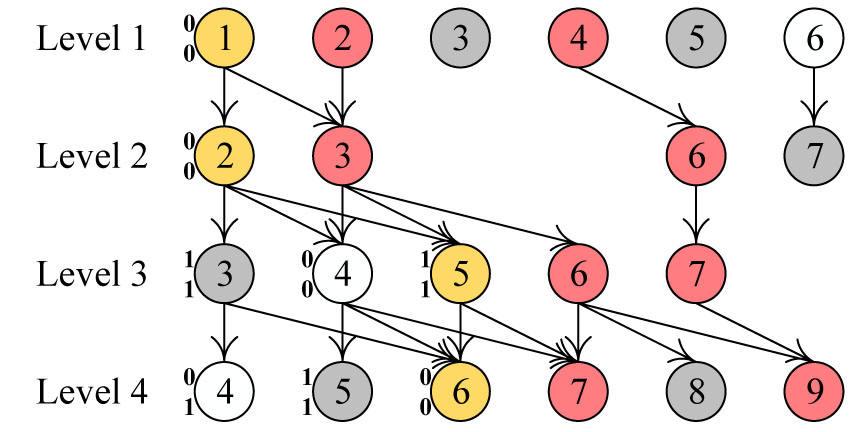}
		\caption{MRD for each node in the subNettree of $n_{1}^{1}$}
		\label{nettree4} 	
	\end{figure}
\end{example}

{The main steps of NetNDP are as follows. }

{Step 1: NetNDP uses Algorithm 1 to create a local approximate Nettree at first. }

{Step 2: NetNDP adopts Algorithm 2 to determine the max root which can reach the absolute leaves under the condition of the global constraint. }

{Step 3: If we can obtain the rightmost absolute leaf of the max root, then we employs Algorithm 3 to get the rightmost occurrence using the rightmost parent strategy with the rightmost absolute leaf. }

{Step 4: We iterate Steps 2 and 3 until no new occurrences are found.}

The NetNDP algorithm is shown in Algorithm 4.

	\begin{algorithm}[t] 
		\caption{NetNDP} 
		\hspace*{0.02in} {\bf Input:} 
		sequence $S$, pattern $P$, local threshold $\delta$, global threshold $\gamma$\\
		\hspace*{0.02in} {\bf Output:} 
		The maximum nonoverlapping $(\delta, \gamma)$-approximate occurrence set $OCC$
		\begin{algorithmic}[1]
\State Use Algorithm 1 to create a local approximate Nettree;
\For{$r =$ the number of roots to 1 step -1} 
	\State Use Algorithm 2 to get the rightmost absolute leaf $ral$;
	\If{$lal ==$ NULL}
		\State Use Algorithm 3 to get the rightmost occurrence $occ$;
		\State $OCC  \leftarrow  OCC \cup occ$;
		\State Delete $occ$;
	\EndIf
\EndFor
\State return $OCC$;	
		\end{algorithmic}
	\end{algorithm}

	\section{Experimental Results and Performance Analysis}
	\label{experimental}
	\subsection{Experimental Environment and Datasets}
	\label{experimental1}
	In this paper, we  focus on the problem of NDP, and propose the NetNDP algorithm. To validate the performance of NetNDP, we select the following competitive algorithms.

(1) INSGrow-appro, NETLAP-$(\delta, \gamma)$, and NETASPNO-$(\delta, \gamma)$:  INSGrow \cite{Ding2009Efficient}, NETLAP-Best \cite{Wu2017Strict} and NETASPNO \cite{Wu2018NETASPNO} are state-of-the-art algorithms for nonoverlapping pattern matching, but they do not support $(\delta, \gamma)$-approximate matching. We therefore  propose INSGrow-appro, NETLAP-$(\delta, \gamma)$, and NETAS- PNO-$(\delta, \gamma)$ to improve these algorithms, respectively.

(2) NetNDP-nonp: To verify the efficiency of our pruning strategy, we also develop an algorithm, called NetNDP-nonp, which does not prune invalid nodes and parent-child relationships. 

(3) NetDAP \cite{wu2020APIN}: To verify the NDP performance, we also select NetDAP \cite{wu2020APIN} as a competitive algorithm which is $(\delta, \gamma)$ - approximate pattern matching under no special condition. 

All of the algorithms are run on a computer with an Intel(R) Core(TM) i5-7200U, 2.50GHz CPU, 8.00GB RAM, and a Windows10 operating system. The compiling environment is VC++ 6.0.
	

{To verify the efficiency of NetNDP, we use eight real protein sequences ($S1\sim S8$), which can be downloaded from https://www.uniprot.org/uniparc/. Protein sequences are composed of 20 different letters (amino acids). Proteins with similar patterns often have similar functions. However, protein sequences may mutate. Therefore, if we know a sequence pattern, we can find similar pattern in the new protein sequences by approximate pattern matching with gap constraints, so as to further study and confirm the functional structure of the new protein sequence. To verify the matching effect of the $(\delta, \gamma)$-distance, we select two WormsTwoClass time series ($S9\sim S10$), which can be downloaded from https://www.cs. ucr.edu/$\sim$eamonn/time\_series\_data/. The dataset describes the trajectory of Caenorhabditis elegans. On average, the length of all sequences is 900. Although the accuracy of the data set is very high, there are also some errors. According to the time series of the first Eigenworm and the known trajectory pattern, we get similar time series by approximate pattern matching, which is easy to find similar types of worms. We apply SAX \cite{Lin2007Experiencing} (https://cs.gmu.edu/ $\sim$jessica/sax.htm) to symbolize the time series as character sequences ($A \sim T$). A description of the datasets is given in Table 2.}
	
	Table 3 shows the patterns $P1$ to $P8$ used to evaluate the experimental performance. Patterns $P1$ and $P2$ are two random patterns. To verify the influence of the length of the pattern on the experimental results, we set the patterns $P3$ to $P5$ as the same gap constraints, and gradually increase the pattern length. To verify the influence of the gap constraints on the experimental results, we set the patterns $P6$ to $P8$ as the same pattern length, and gradually increase the gap constraint. To verify the matching effect of the $(\delta, \gamma)$-distance, we add the patterns $P9$ to $P10$, and their trends are shown in Figure \ref{trend}.
	
		\begin{table}[t] 
		\centering
		\caption{Description of the datasets}
		\label{sequences}       
		\begin{tabular}{cccc}
			\hline\noalign{\smallskip}
			Sequence & Identification &	From & Length  \\
			\noalign{\smallskip}\hline\noalign{\smallskip}
			$S1$ & UPI000E62E3D8 & Schistosoma mansoni & 2578\\
			$S2$ & UPI000E62F0E8 & Ascaris suum & 3081\\
			$S3$ & UPI000EC610A6 & Corallococcus sp. CA051B & 3895\\
			$S4$ & UPI000E6F28B3 & Xiphophorus maculatus & 4260\\
			$S5$ & UPI0000F516B0 & Mus musculus & 4632\\
			$S6$ & UPI00078AE9A0 & Drosophila simulans & 4864\\
			$S7$ & UPI00090EABF2 & Xenopus laevis & 5606\\
			$S8$ & UPI000E641A62 & Ascaris suum & 6638\\
			$S9$ & Train79 & WormsTwoClass & 900\\
			$S10$ &	Train85 & WormsTwoClass & 900\\
			\noalign{\smallskip}\hline
		\end{tabular}
	\end{table}
	
	\begin{table}[t] 
		\centering
		\caption{Patterns}
		\label{patterns}       
		\begin{tabular}{ccc}
			\hline\noalign{\smallskip}
			Name & Pattern & Length  \\
			\noalign{\smallskip}\hline\noalign{\smallskip}
			$P1$ & V[1,5]L[1,7]S[4,9]L & 4 \\
			$P2$ & L[1,7]T[0,6]S[3,8]L[2,7] & 5 \\
			$P3$ & E[0,9]L[0,9]S[0,9]E[0,9]L & 5 \\
			$P4$ & E[0,9]L[0,9]S[0,9]E[0,9]L[0,9]S[0,9]E & 7 \\
			$P5$ & E[0,9]L[0,9]S[0,9]E[0,9]L[0,9]S[0,9]E[0,9]L & 9 \\
			$P6$ & Q[1,7]E[1,7]L[1,7]E[1,7]L[1,7]N & 6 \\
			$P7$ & Q[1,8]E[1,8]L[1,8]E[1,8]L[1,8]N & 6 \\
			$P8$ & Q[1,10]E[1,10]L[1,10]E[1,10]L[1,10]N & 6 \\
			$P9$ & P[0,6]M[0,6]D[0,6]L[0,6]Q & 5 \\
			$P10$ & F[0,6]B[0,6]J[0,6]Q[0,6]E[0,6]B & 6 \\
			\noalign{\smallskip}\hline
		\end{tabular}
	\end{table}

\begin{figure} 
	\centering
	\begin{minipage}[t]{0.5\linewidth}	
		\includegraphics[width=2.3in]{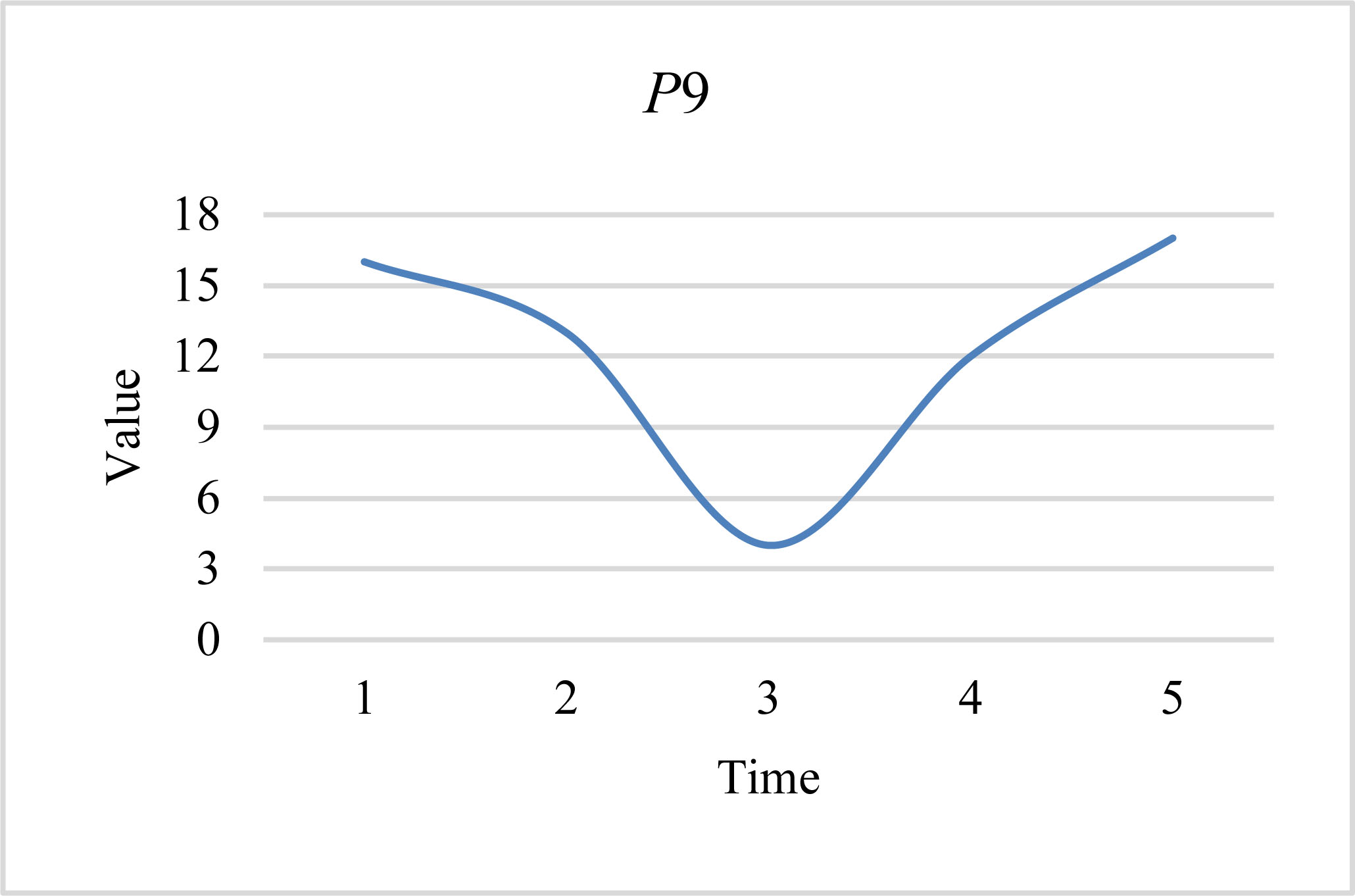}
	\end{minipage}%
	\begin{minipage}[t]{0.5\linewidth}	
		\includegraphics[width=2.3in]{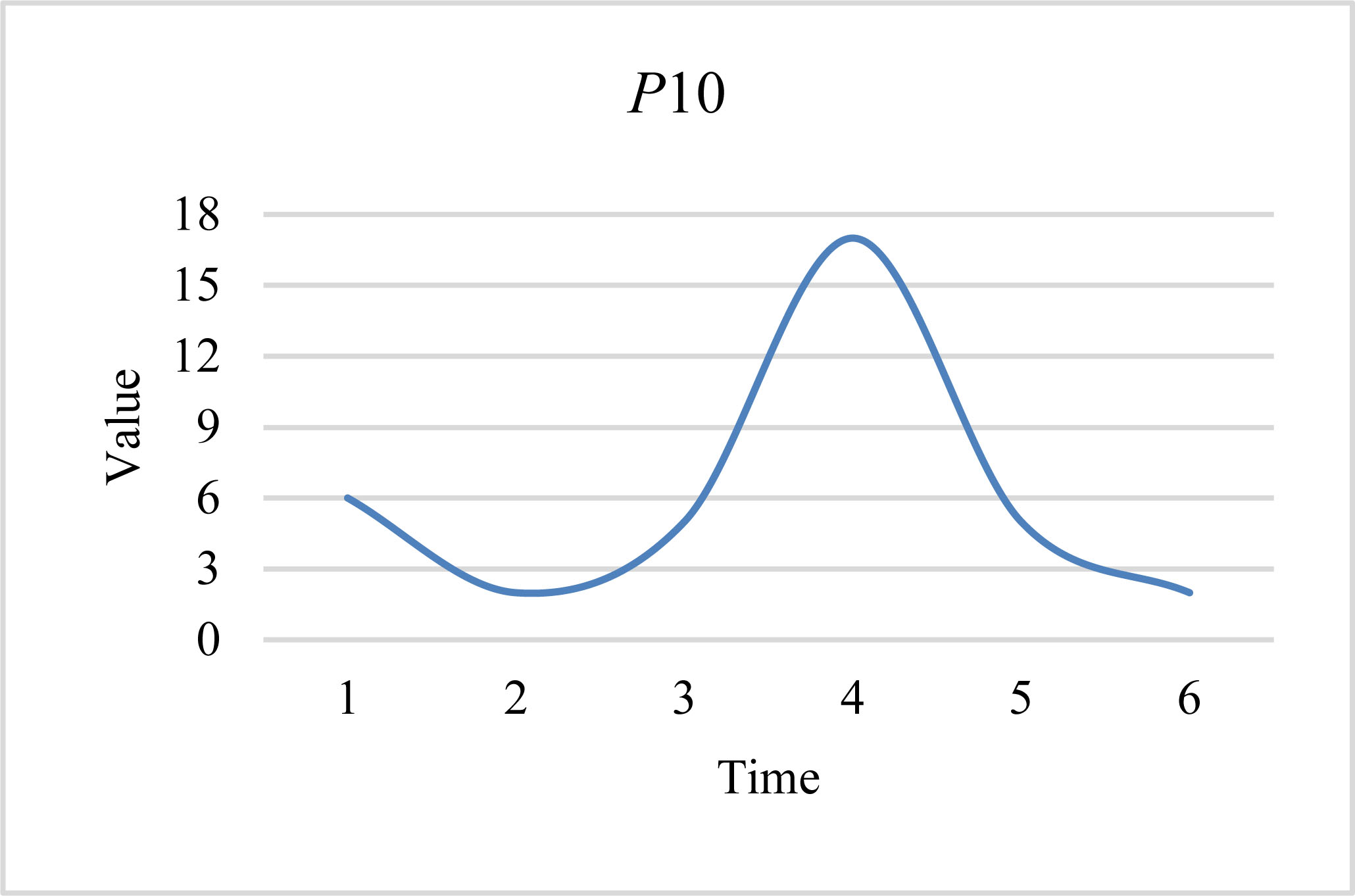}
	\end{minipage}%
	\caption{Trends in $P9\sim P10$}
	\label{trend}
\end{figure}

	\subsection{Efficiency}	\label{Experimental Results and Analysis}
	Since large deviations are not allowed in the matching results, the local threshold $\delta$ should not be set very high. Since the length of pattern $P1$ is four, the global threshold $\gamma$ should also not be set very high. We therefore use values of  $(\delta=1, \gamma=2)$, $(\delta=1, \gamma=3)$, and $(\delta=2, \gamma=2)$. Figures \ref{fig8} and \ref{fig9}  show a comparison of the results and running time for $P1\sim P8$ on $S1\sim S8$ with $(\delta=1, \gamma=2)$.
	
	\begin{figure} 
		\centering 
		\begin{minipage}[t]{0.5\linewidth}	
			\includegraphics[width=2.3in]{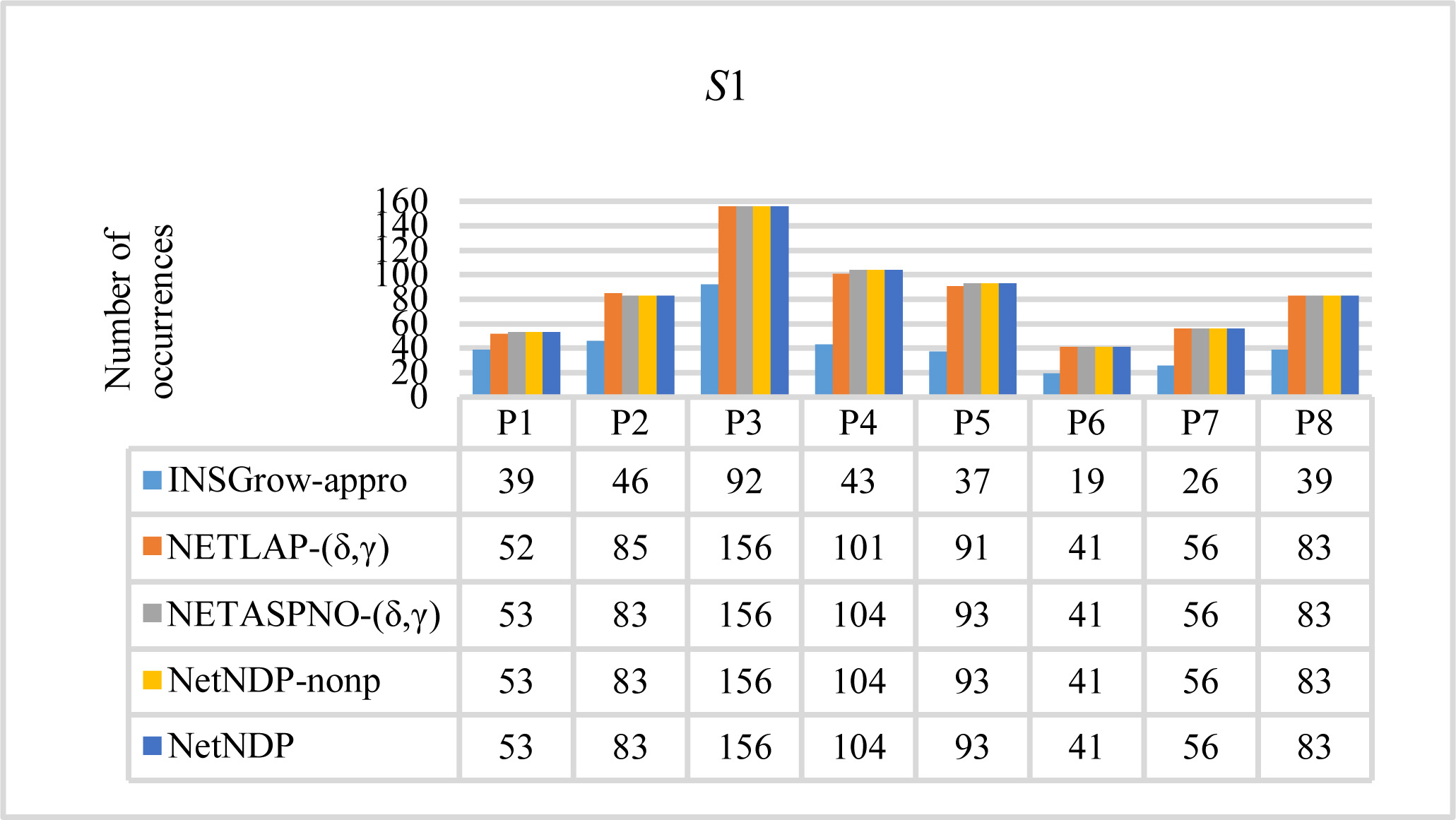}
		\end{minipage}%
		\begin{minipage}[t]{0.5\linewidth}	
			\includegraphics[width=2.3in]{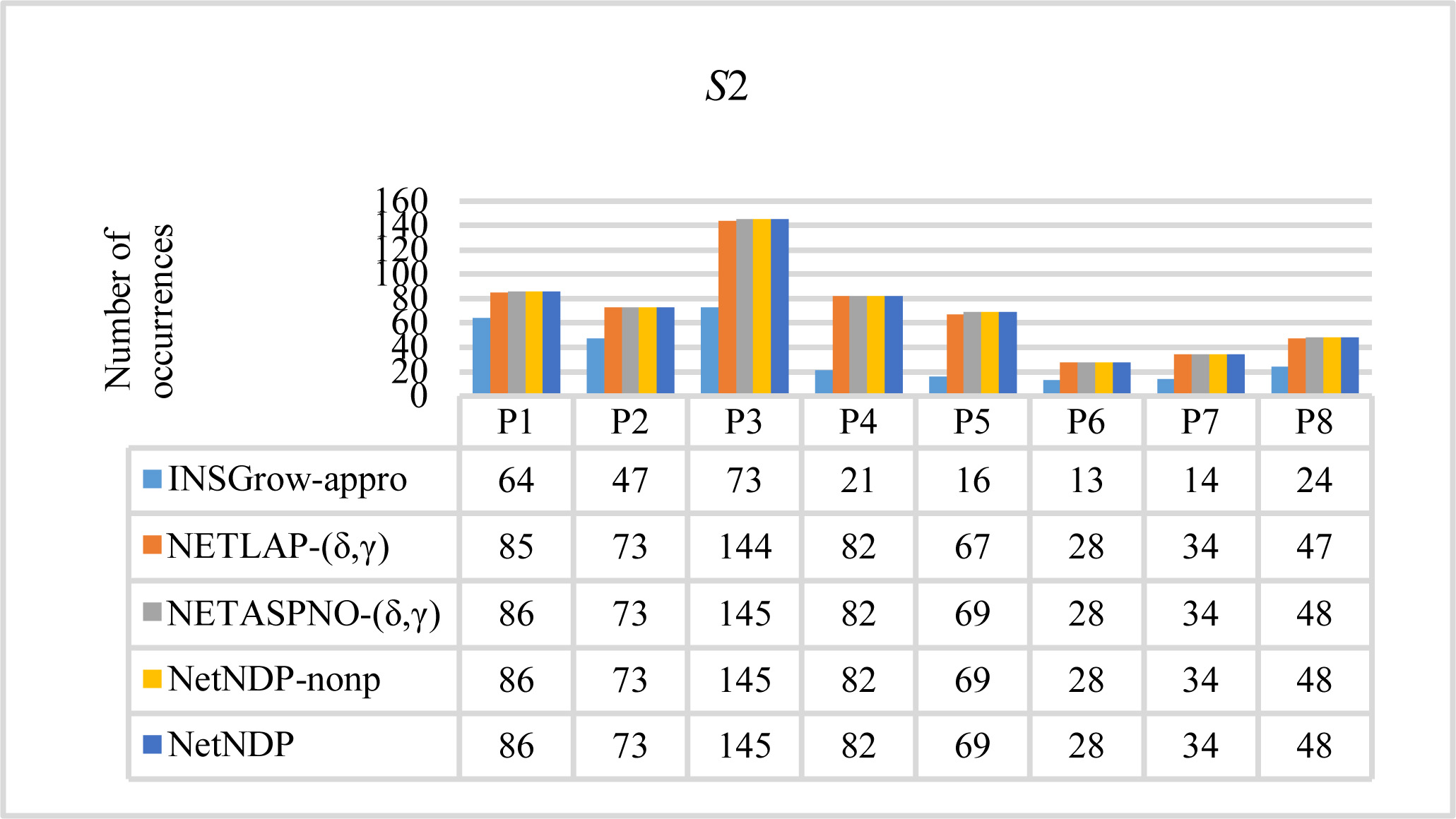}
		\end{minipage}%
		\quad
		\begin{minipage}[t]{0.5\linewidth}	
			\includegraphics[width=2.3in]{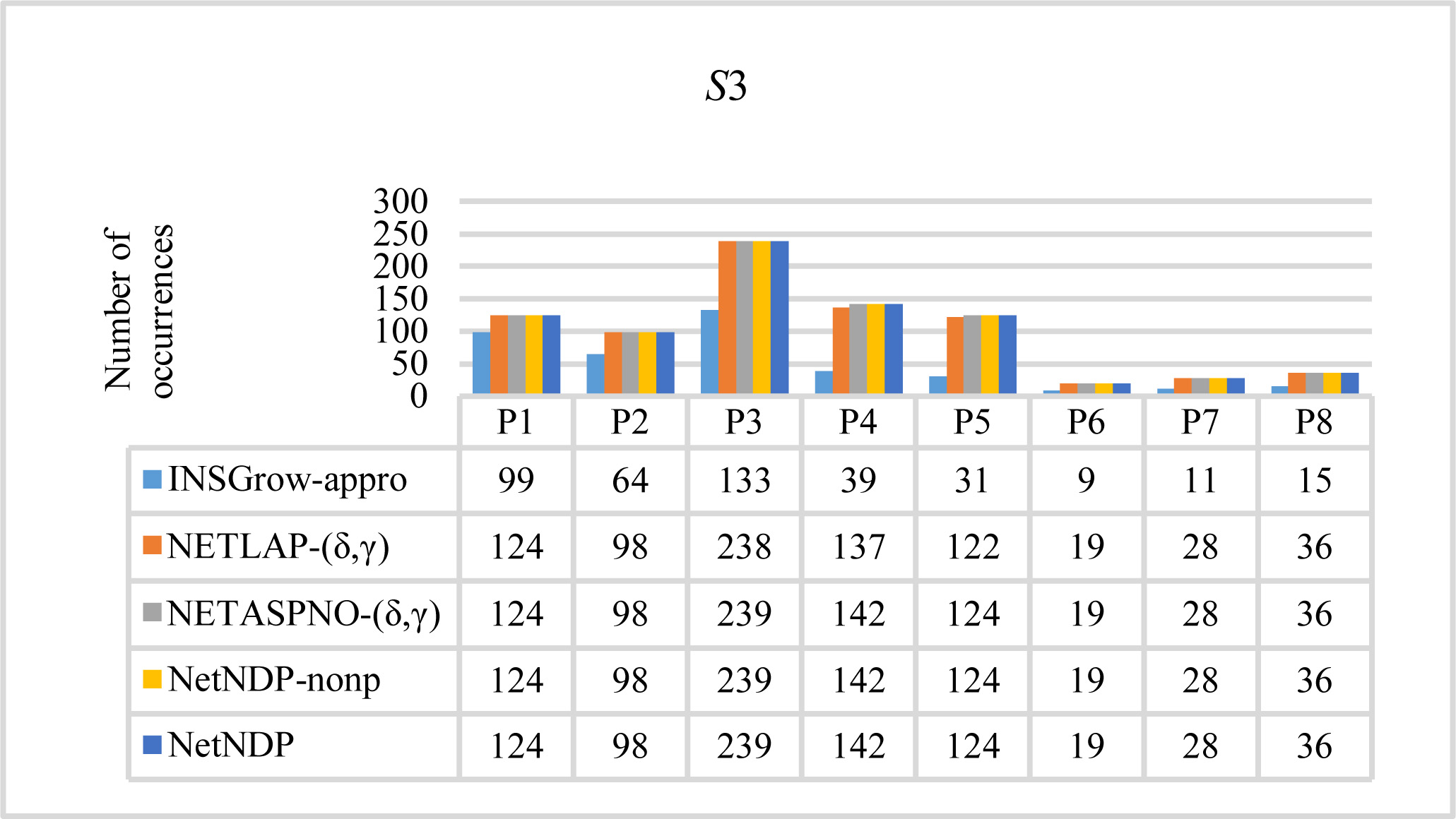}
		\end{minipage}%
		\begin{minipage}[t]{0.5\linewidth}	
			\includegraphics[width=2.3in]{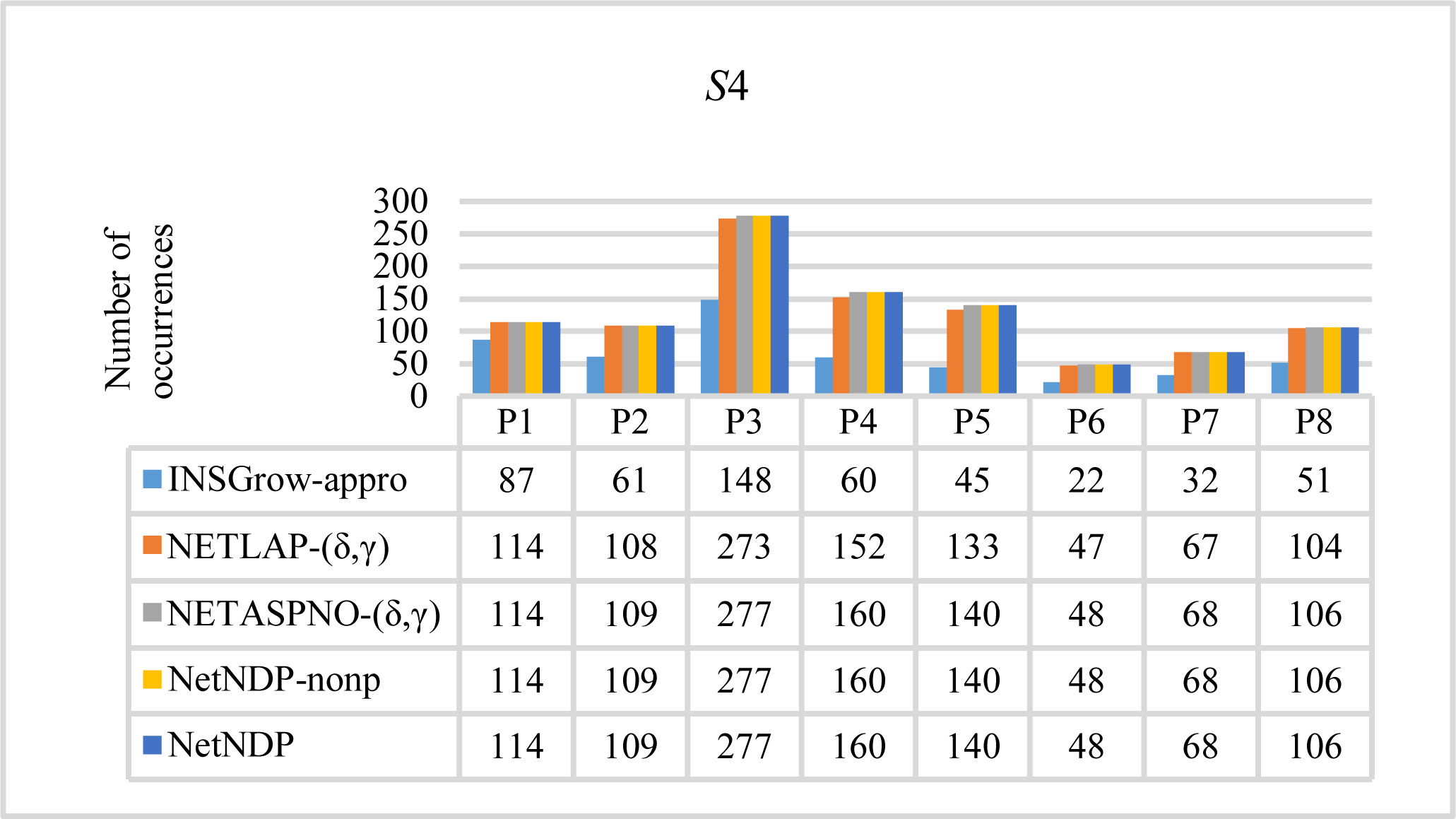}
		\end{minipage}%
		\quad
		\begin{minipage}[t]{0.5\linewidth}	
			\includegraphics[width=2.3in]{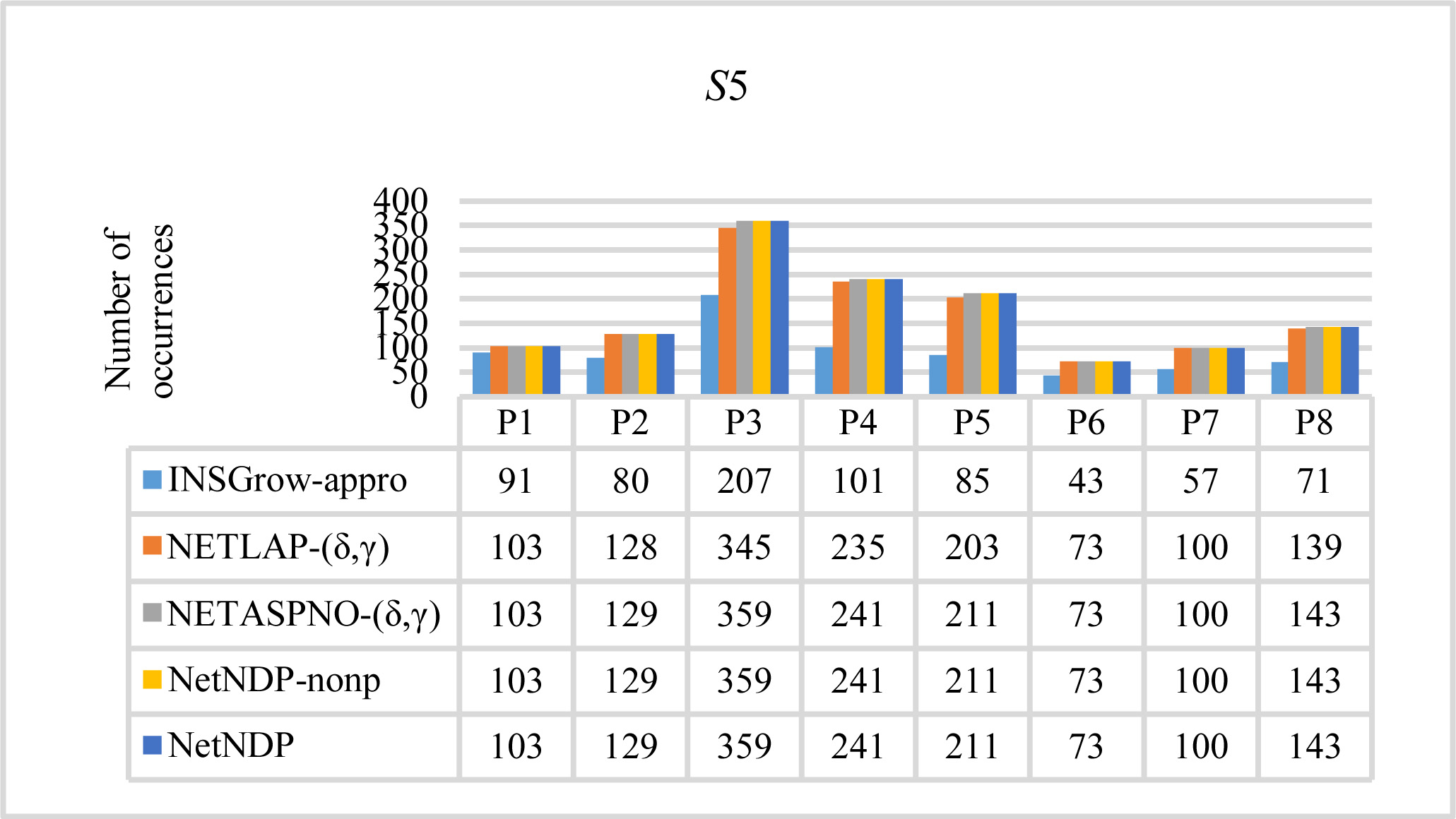}
		\end{minipage}%
		\begin{minipage}[t]{0.5\linewidth}	
			\includegraphics[width=2.3in]{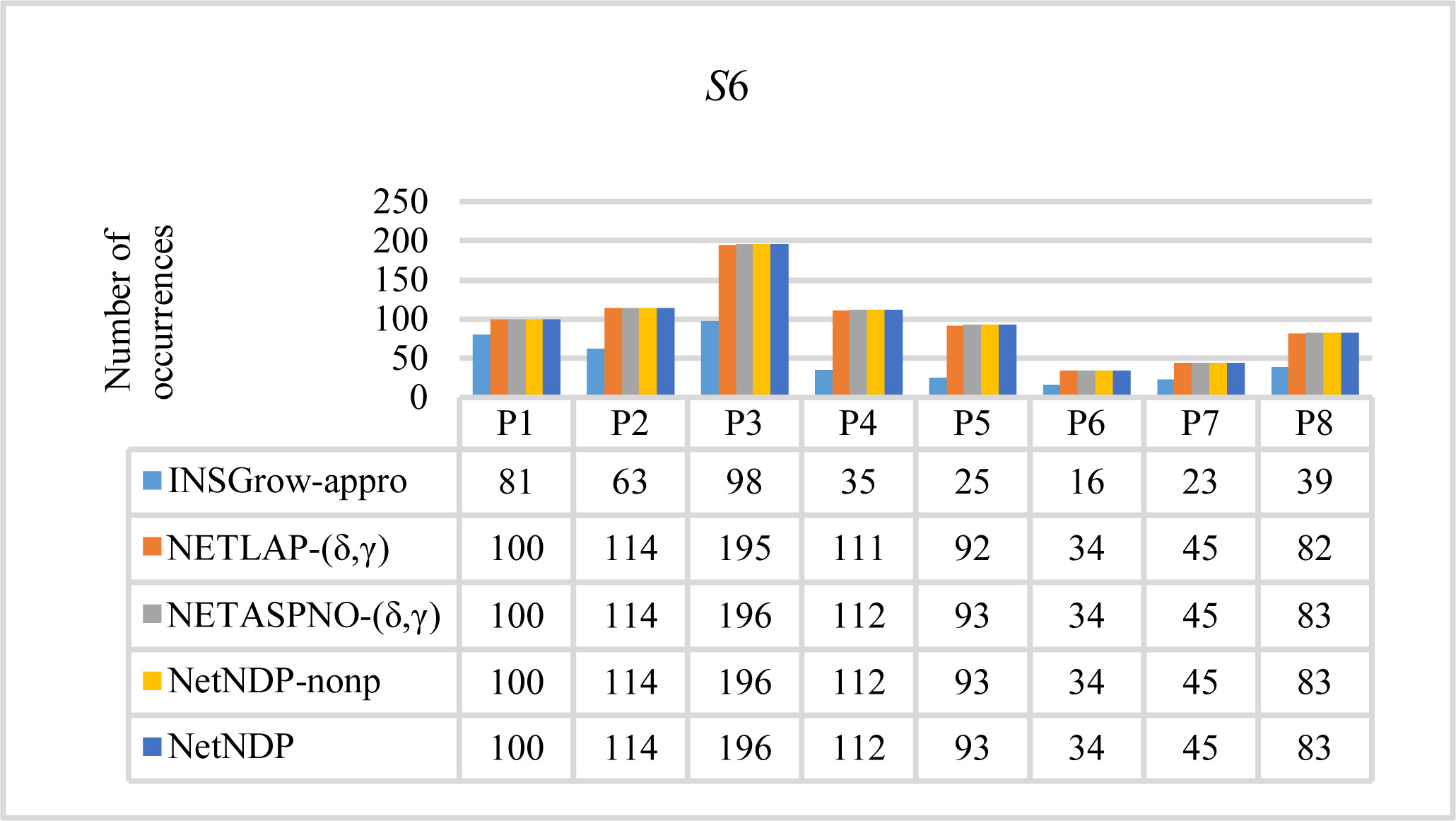}
		\end{minipage}%
		\quad
		\begin{minipage}[t]{0.5\linewidth}	
			\includegraphics[width=2.3in]{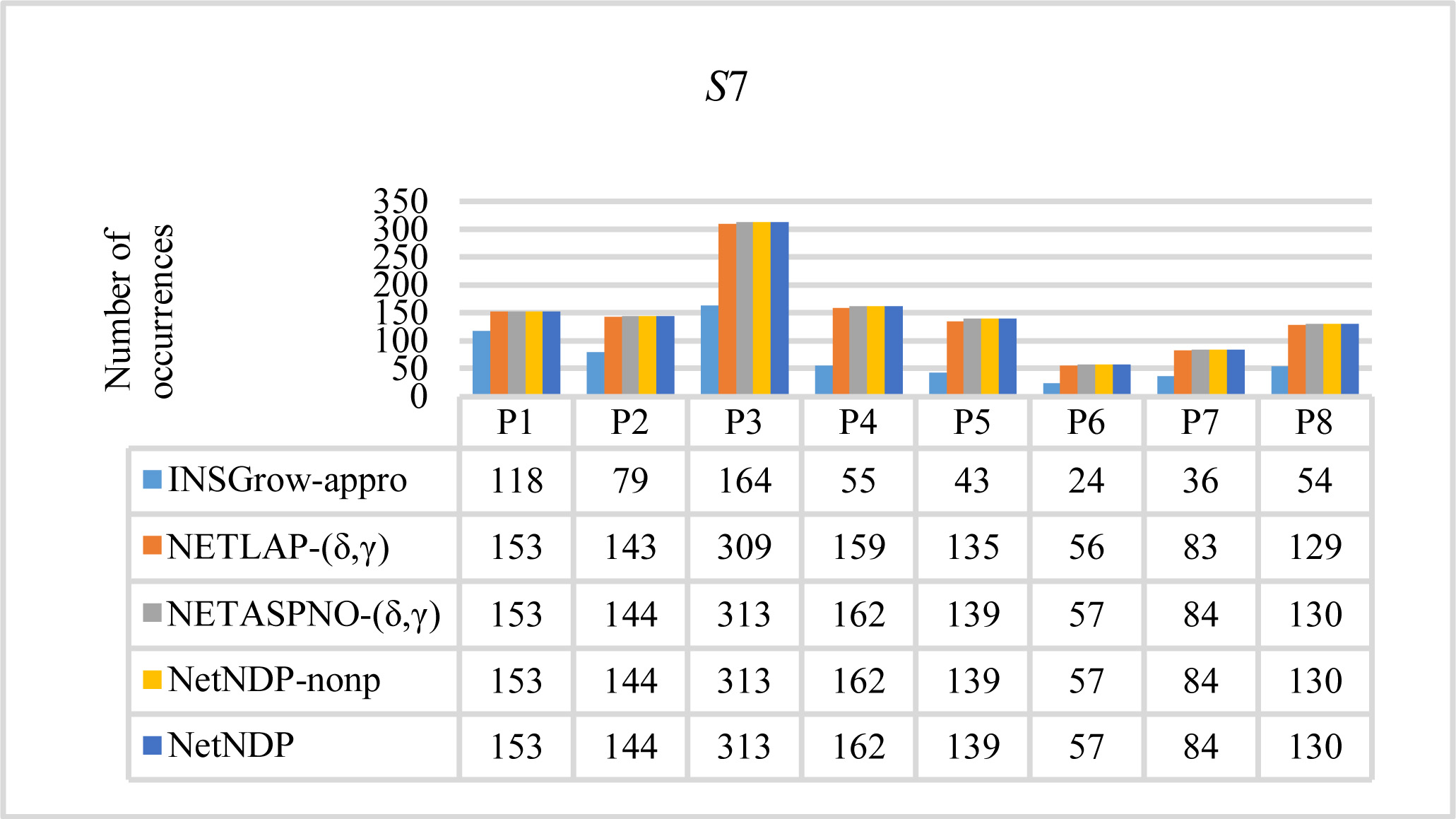}
		\end{minipage}%
		\begin{minipage}[t]{0.5\linewidth}	
			\includegraphics[width=2.3in]{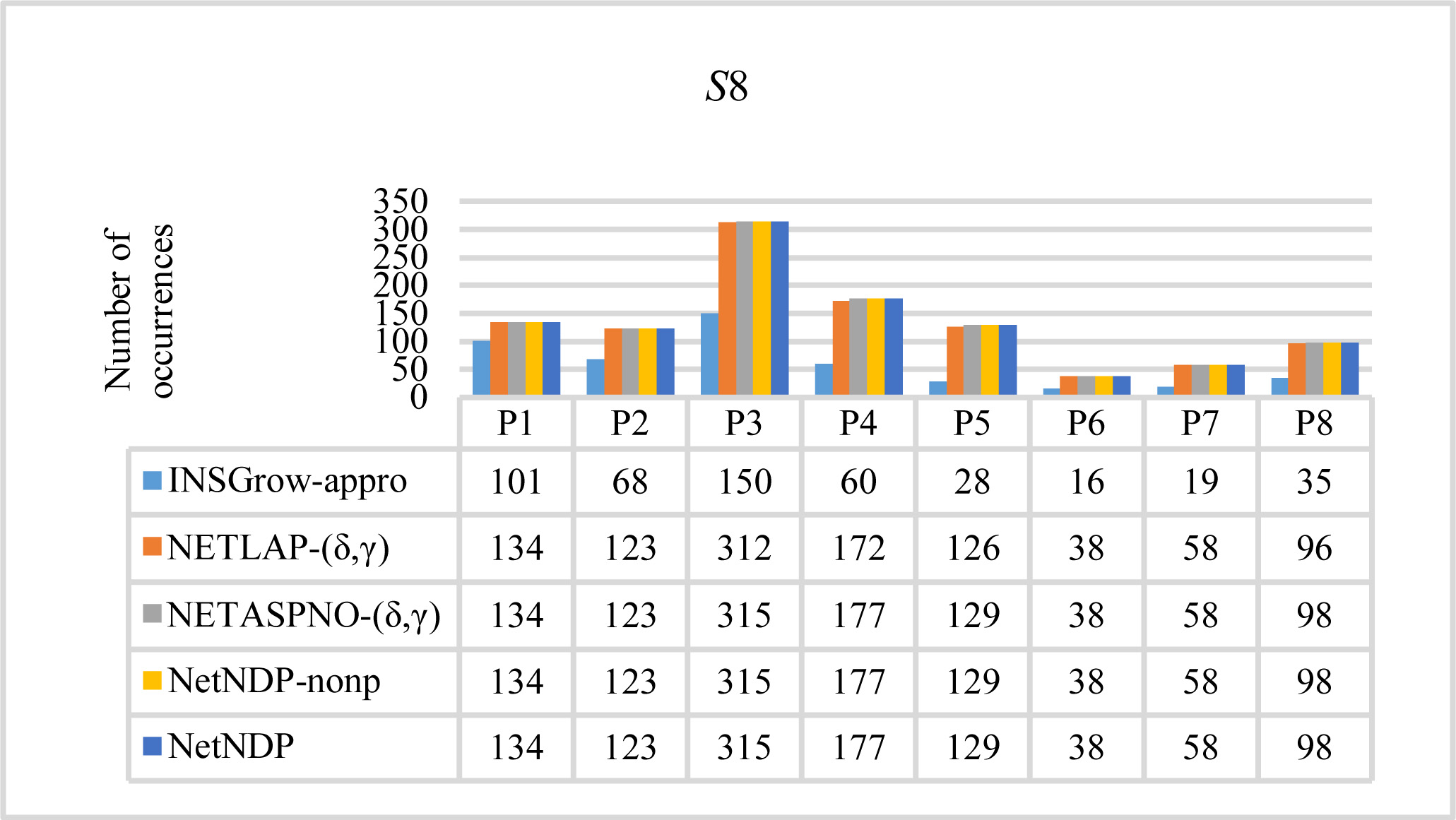}
		\end{minipage}%
		\caption{Comparison of results for $P1\sim  P8$ on $S1\sim S8$ with $(\delta=1, \gamma=2)$} 
		\label{fig8} 
	\end{figure}
	
	\begin{figure} 
		\centering 
		\begin{minipage}[t]{0.5\linewidth}	
			\includegraphics[width=2.3in]{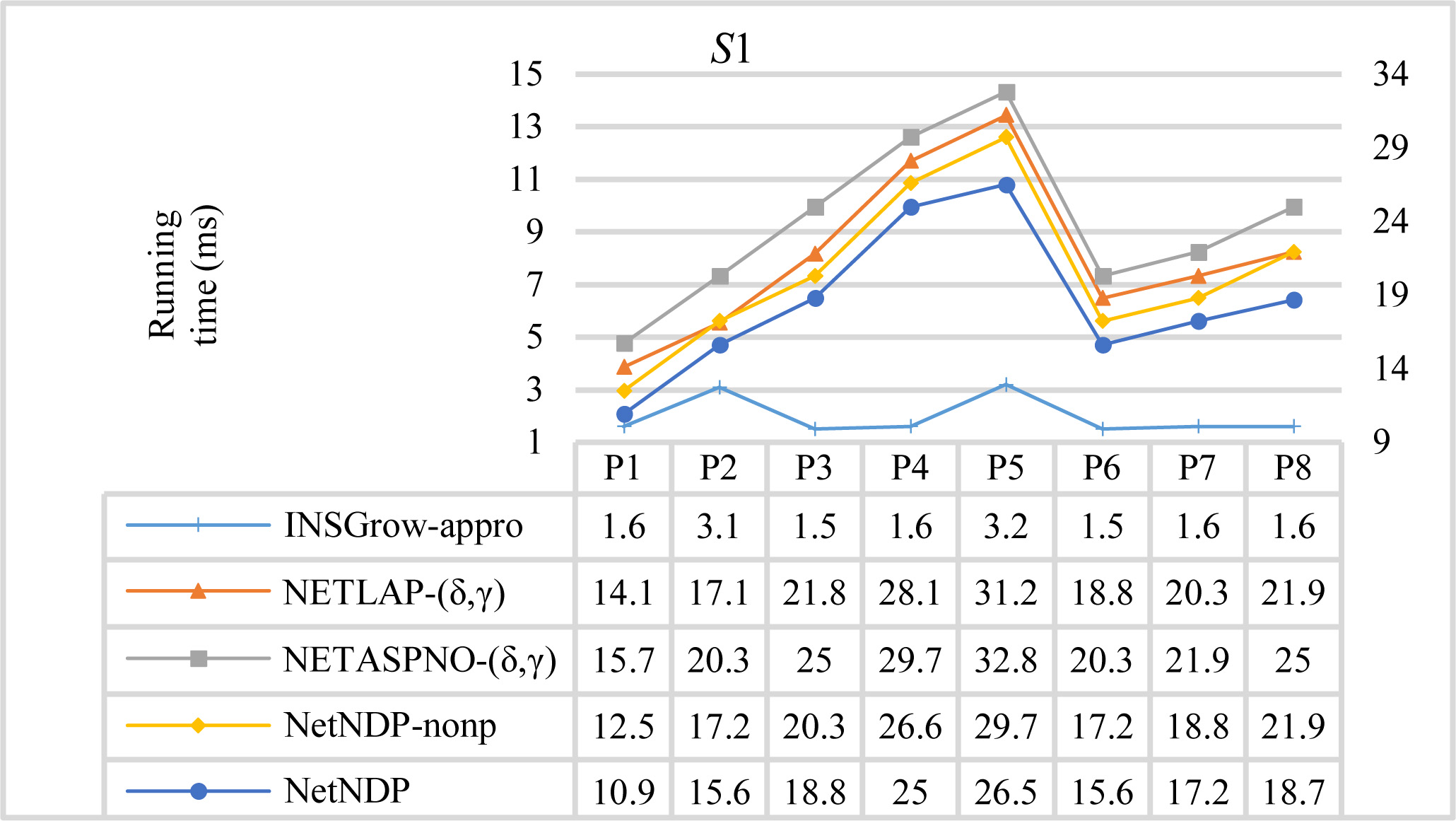}
		\end{minipage}%
		\begin{minipage}[t]{0.5\linewidth}	
			\includegraphics[width=2.3in]{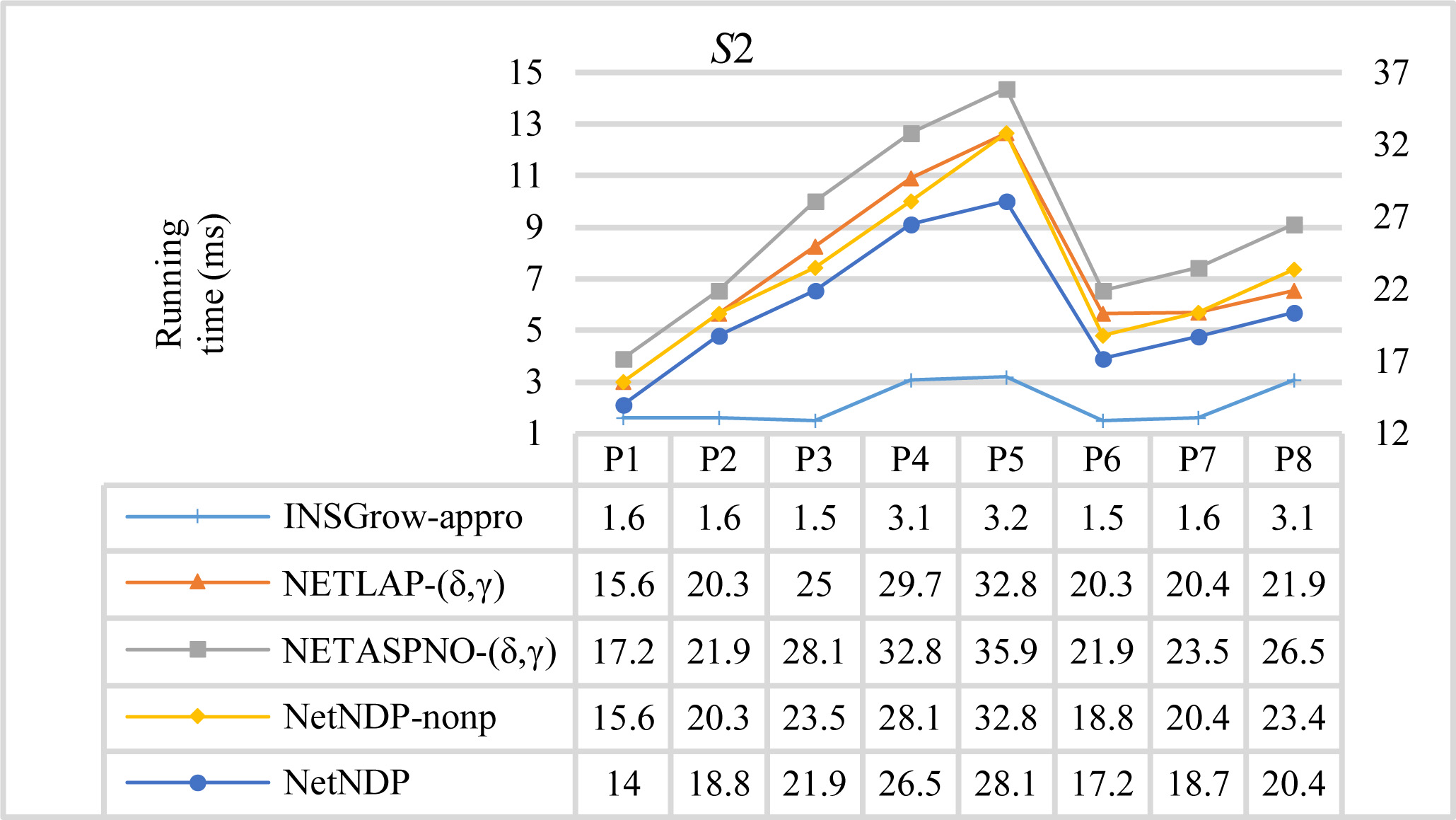}
		\end{minipage}%
		\quad
		\begin{minipage}[t]{0.5\linewidth}	
			\includegraphics[width=2.3in]{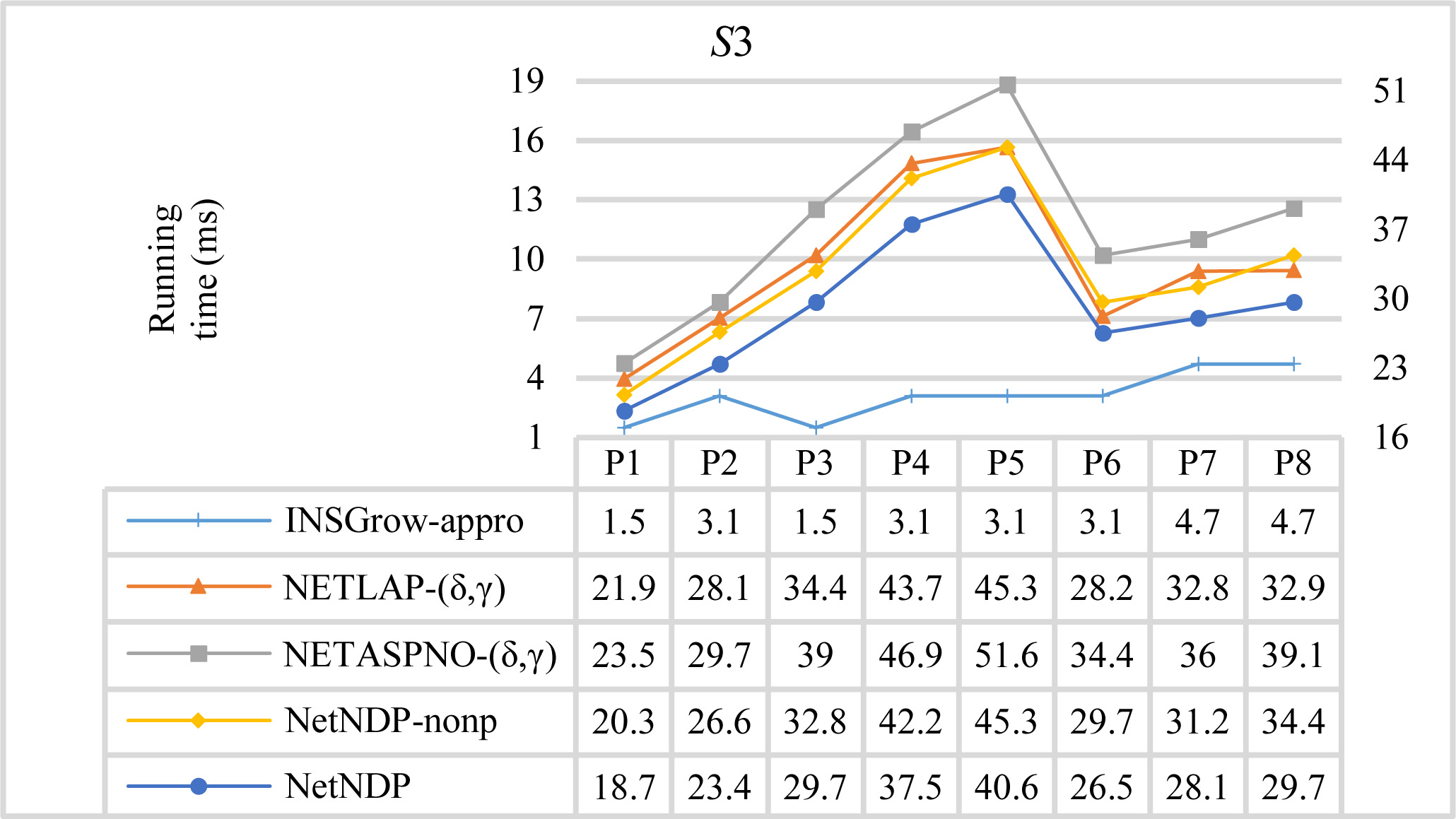}
		\end{minipage}%
		\begin{minipage}[t]{0.5\linewidth}	
			\includegraphics[width=2.3in]{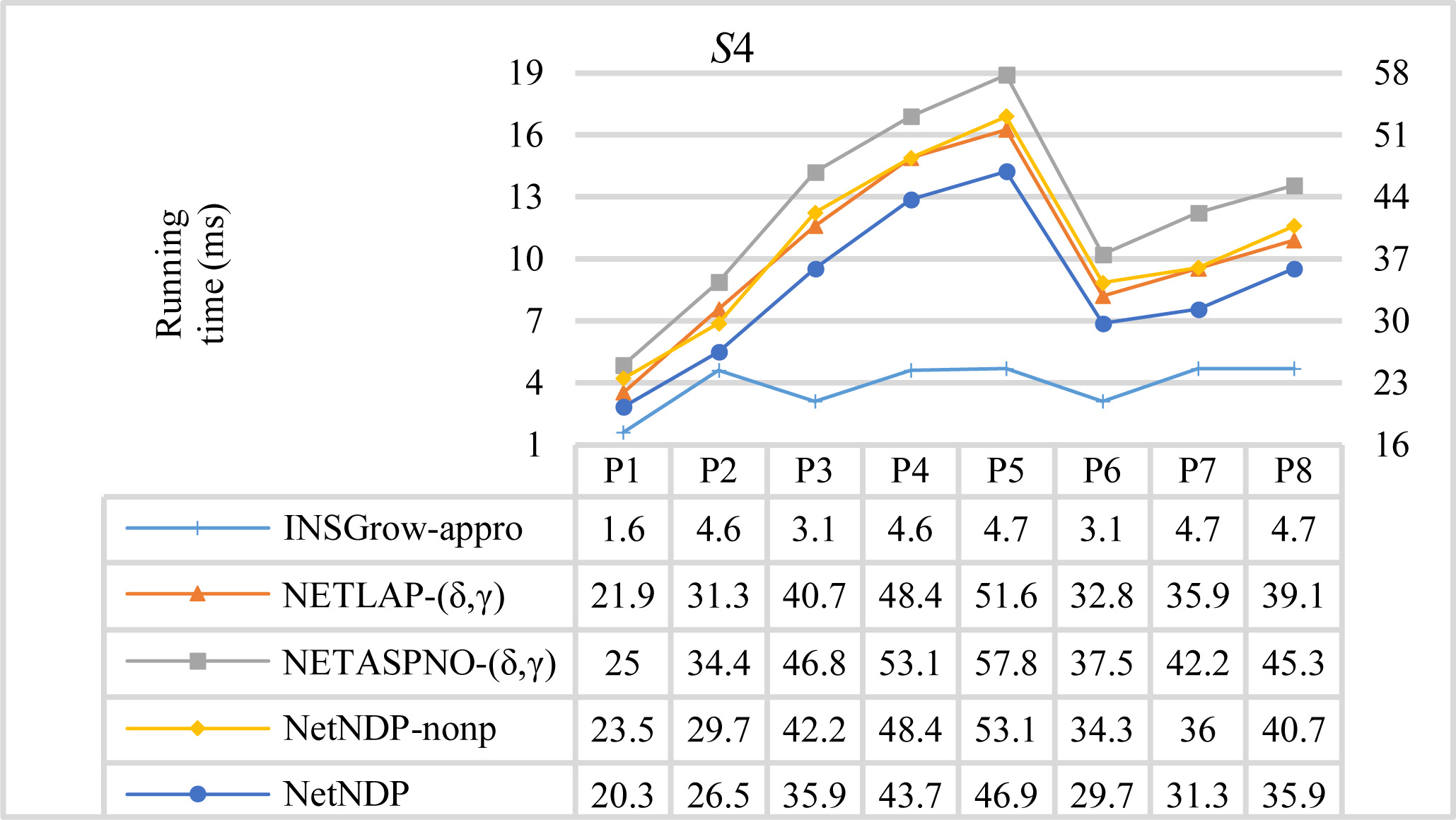}
		\end{minipage}%
		\quad
		\begin{minipage}[t]{0.5\linewidth}	
			\includegraphics[width=2.3in]{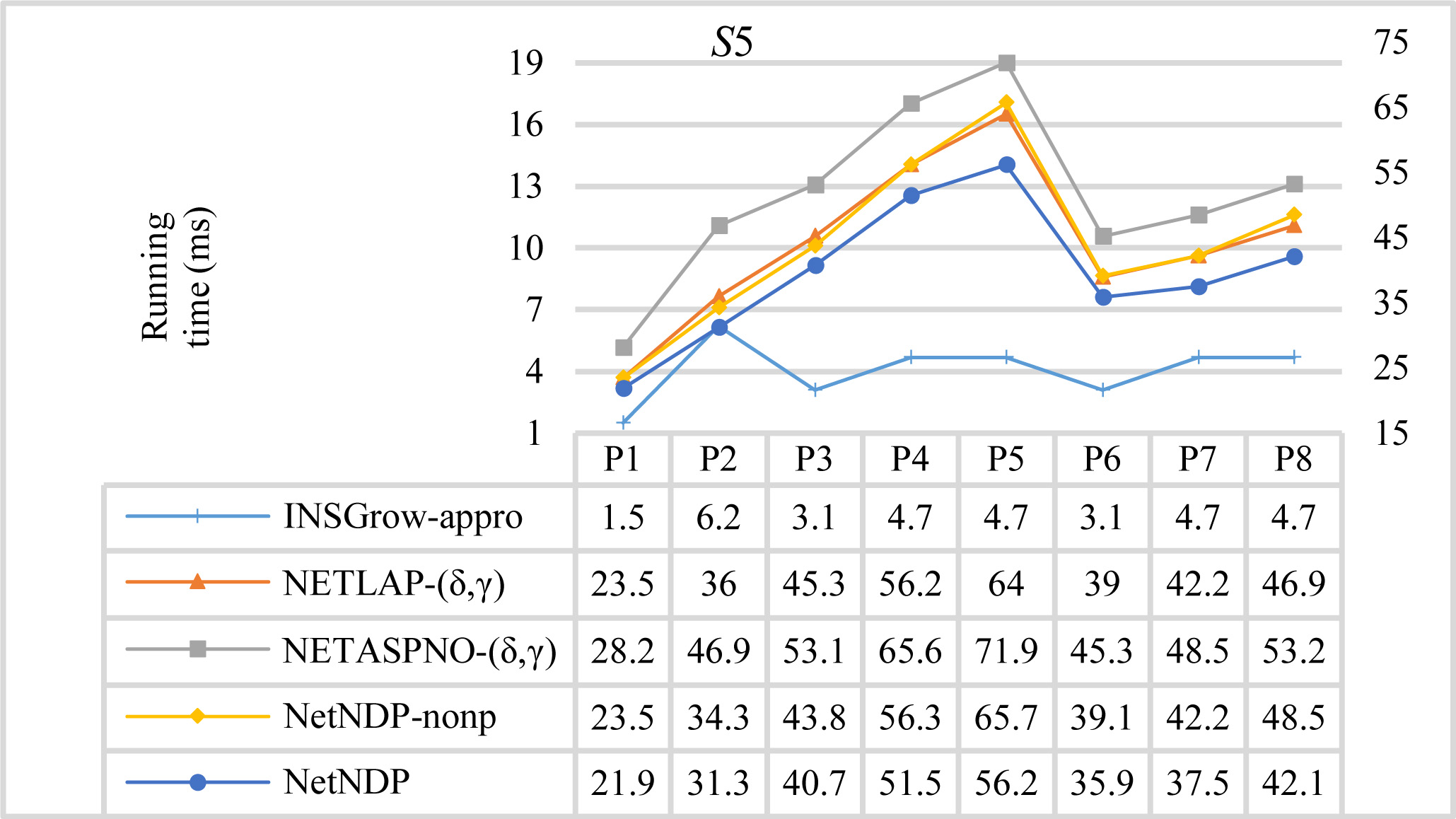}
		\end{minipage}%
		\begin{minipage}[t]{0.5\linewidth}	
			\includegraphics[width=2.3in]{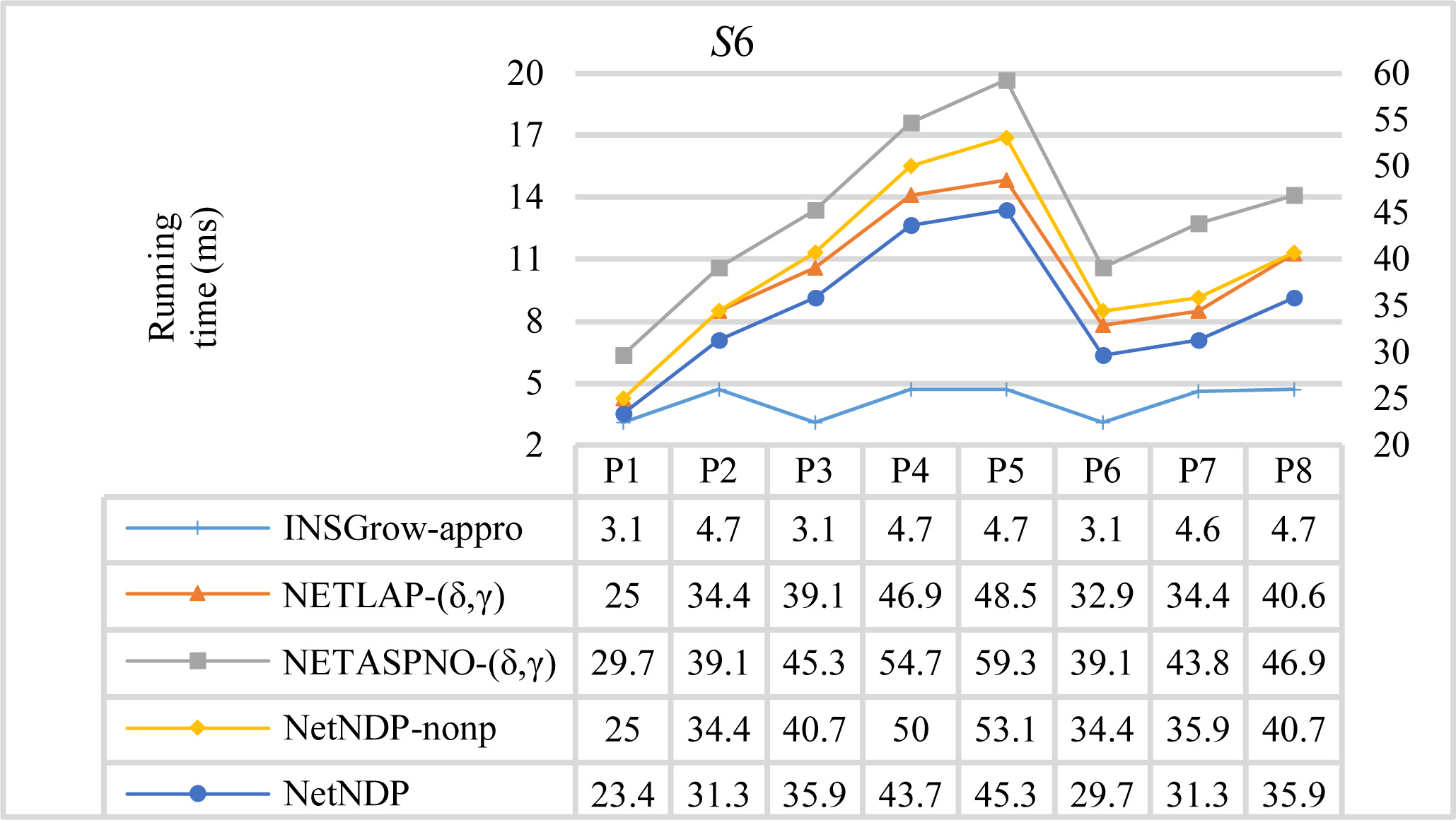}
		\end{minipage}%
		\quad
		\begin{minipage}[t]{0.5\linewidth}	
			\includegraphics[width=2.3in]{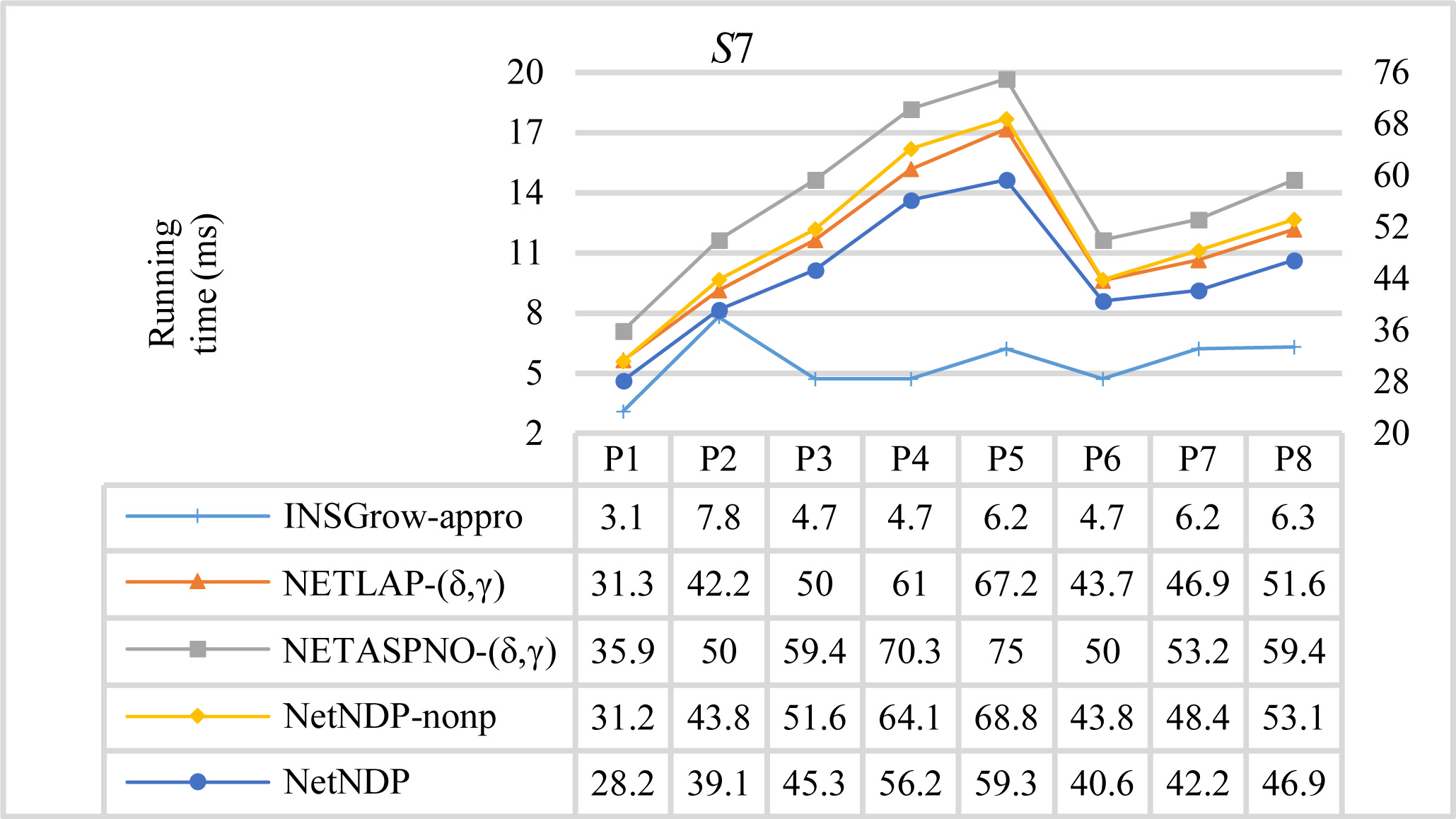}
		\end{minipage}%
		\begin{minipage}[t]{0.5\linewidth}	
			\includegraphics[width=2.3in]{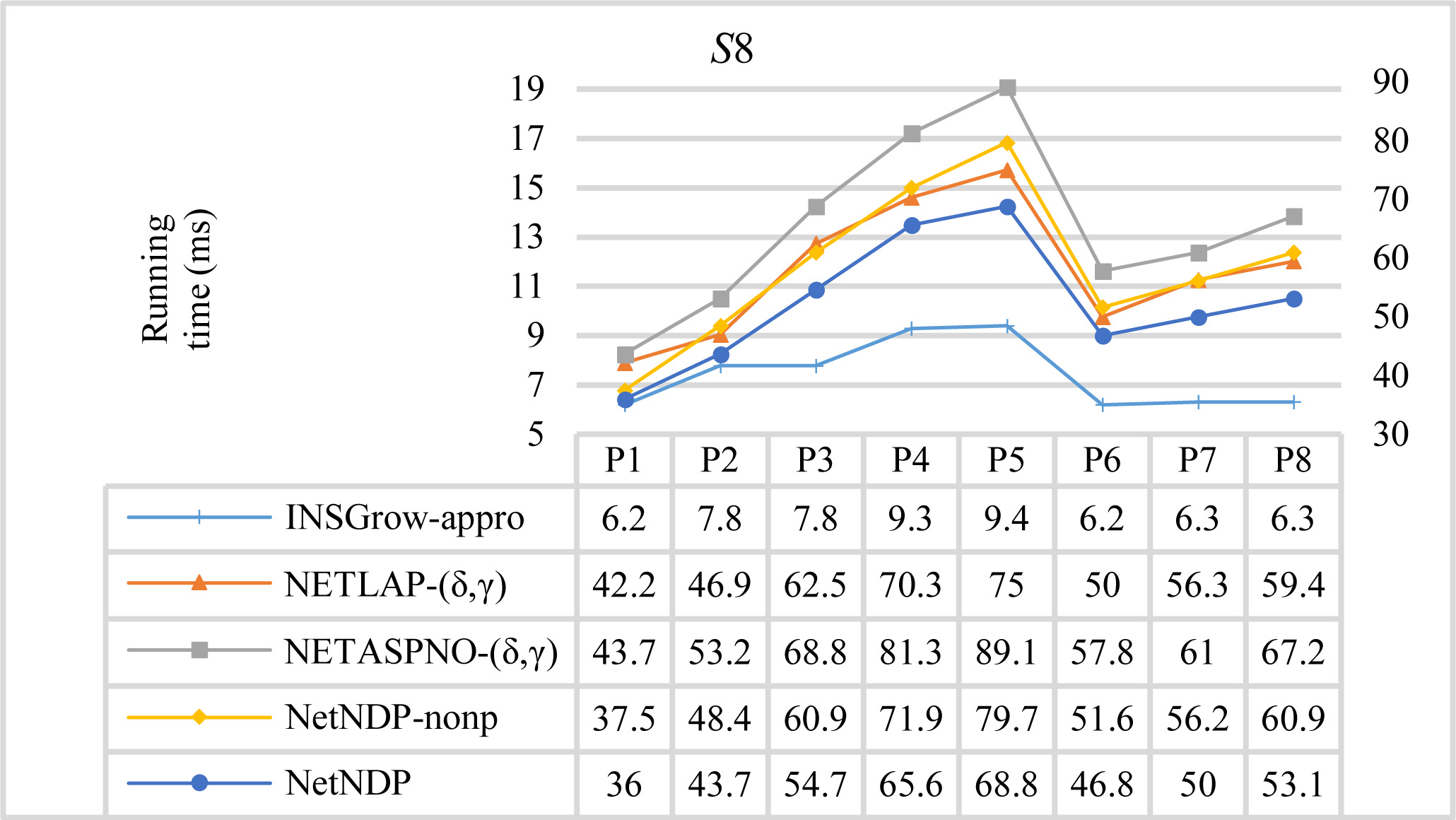}
		\end{minipage}%
		\caption{Comparison of running time for $P1\sim P8$ on $S1\sim S8$ with $(\delta=1, \gamma=2)$} 
		\label{fig9} 
	\end{figure}

	The results indicate the following observations.

	(1) From Figures \ref{fig8}  and \ref{fig9}, we can see that NetNDP outperforms INSGrow-appro. Figure \ref{fig9}  shows that INSGrow-appro is the fastest algorithm, since this algorithm is relatively simple. However, from Figure \ref{fig8}, we know that INSGrow-appro has the lowest number of occurrences. For example, the running time for INSGrow-appro for $P2$ on $S5$ is 6.2ms, which is the fastest of the five algorithms. However, INSGrow-appro only obtains 80 occurrences, while NetNDP finds 129. The reason for this is that INSGrow-appro employs the principle of INSGrow-appro \cite{Ding2009Efficient}, which only retains the position in which the first match succeeds, thus resulting in the loss of many occurrences. INSGrow-appro therefore overlooks many feasible occurrences, meaning that NetNDP has better performance.
	
	(2) NetNDP outperforms NETLAP-$(\delta, \gamma)$. From Figure \ref{fig8}, we can see that NETLAP-$(\delta, \gamma)$ finds slightly fewer occurrences. For example, the number of occurrences found by NetNDP for $P4$ on $S4$ is 160, while NETLAP-$(\delta, \gamma)$ only obtains 152. The main reason for this is that NETLAP-$(\delta, \gamma)$ employs the principle of NETLAP-Best \cite{Wu2017Strict}, which iteratively finds the rightmost occurrence from the max absolute leaf by finding the rightmost parent, and deletes the occurrence and related invalid nodes. In a similar way to NETLAP-Best, NETLAP-$(\delta, \gamma)$ does not prejudge the rightmost absolute leaf of each root, resulting in fewer occurrences. The illustrative example given in Section 4.2.2 shows that NETLAP-$(\delta, \gamma)$ only finds $<$1, 3, 6, 8$>$ and $<$4, 6, 7, 9$>$, while NetNDP finds three occurrences: $<$1, 2, 5, 6$>$, $<$2, 3, 6, 7$>$, and $<$4, 6, 7, 9$>$. As can been seen from Figure \ref{fig9}, NetNDP is also faster than NETLAP-$(\delta, \gamma)$. For example, the running time for NetNDP is 46.9ms for $P5$ on $S4$, while the running time for NETLAP-$(\delta, \gamma)$ is 51.6ms. NETLAP-$(\delta, \gamma)$ and NetNDP have the same time complexity of $O(n \times m^{2} \times W)$. The reason why NETLAP-$(\delta, \gamma)$ is slower is that it needs to prune invalid nodes associated with occurrences. NetNDP does not need to find these invalid nodes, and therefore NetNDP has better performance than NETLAP-$(\delta, \gamma)$.

	(3) NetNDP outperforms both NETASPNO-$(\delta, \gamma)$ and NetNDP-nonp. From Figure \ref{fig8}, we can see that these three algorithms obtain the same results, and give better performance than the first two algorithms. However, from Figure \ref{fig9}, we know that NetNDP has the best time efficiency. For example, although all three algorithms find 112 occurrences for $P4$ on $S6$, the running time for NetNDP, NETASPNO-$(\delta, \gamma)$ and NetNDP-nonp are 43.7, 54.7, and 50 ms, respectively. The reason for this is because NETASPNO-$(\delta, \gamma)$ and NetNDP-nonp employ the same strategy as NetNDP to search for the occurrences, i.e. they find the rightmost absolute leaf of the max root, search for the rightmost occurrence from the rightmost absolute leaf, and delete the occurrence, meaning that they obtain the same results. The difference is that NETASPNO-$(\delta, \gamma)$ employs the principle of NETASPNO \cite{Wu2018NETASPNO}, which calculates the number of root paths with a $\gamma$-distance of $d$ $(0 \leq d \leq \gamma)$ for each node. NETASPNO judges whether or not a node has root paths that satisfy the distance constraint, and prunes invalid nodes and parent-child relationships via the number of root paths. Thus, the time complexity of NETASPNO-$(\delta, \gamma)$ is $O(n \times m^{2} \times W \times \gamma)$, meaning that it has the longest running time. NetNDP-nonp and NetNDP have the same time complexity of $O(n \times m^{2} \times W)$, but NetNDP-nonp does not prune invalid nodes and parent-child relationships. Figure \ref{fig10} further shows the total numbers of nodes and parent-child relationships for NetNDP-nonp and NetNDP. 

	\begin{figure} 
		\centering
		\begin{minipage}[t]{0.5\linewidth}	
			\includegraphics[width=2.3in]{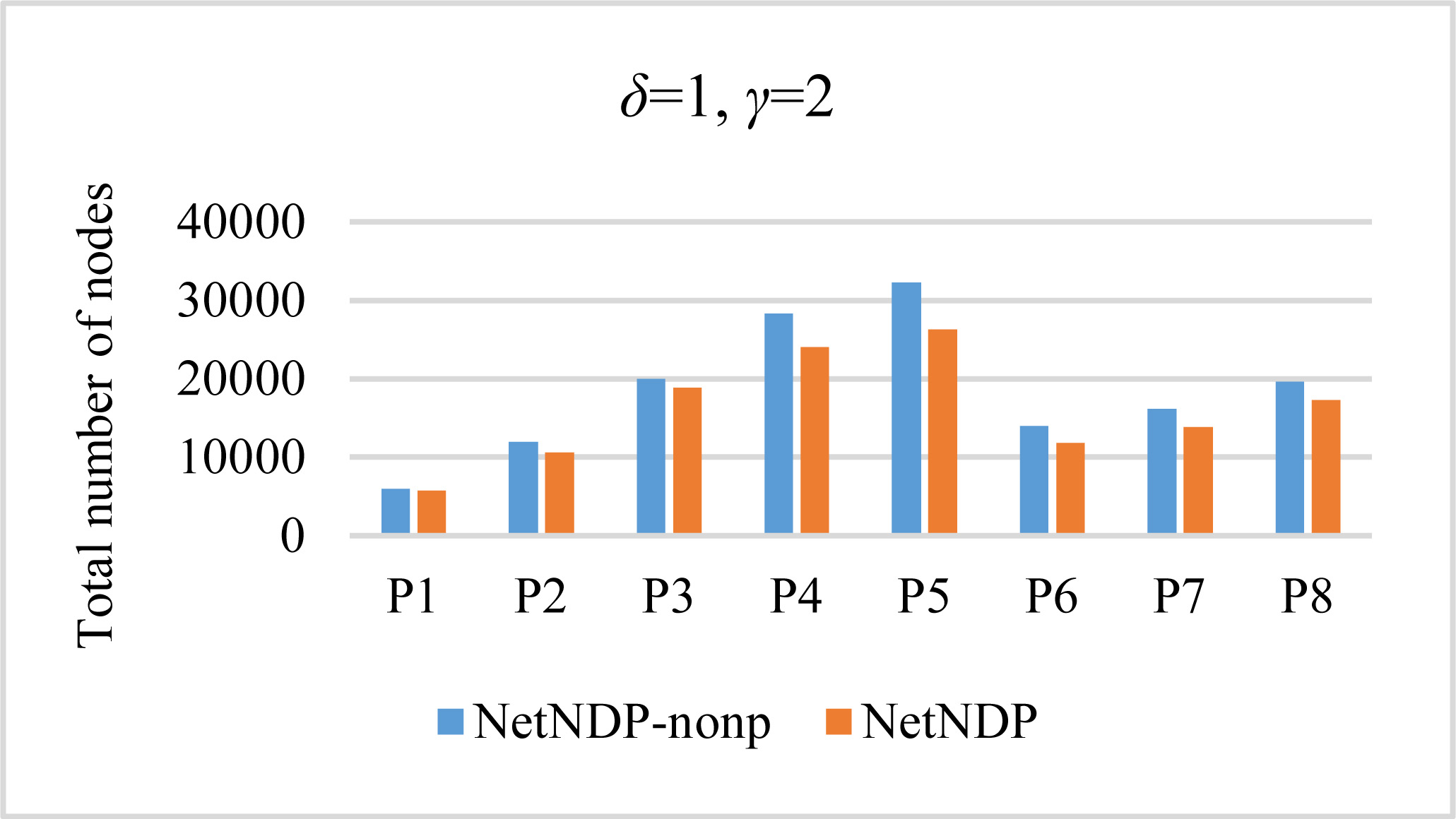}
		\end{minipage}%
		\begin{minipage}[t]{0.5\linewidth}	
			\includegraphics[width=2.3in]{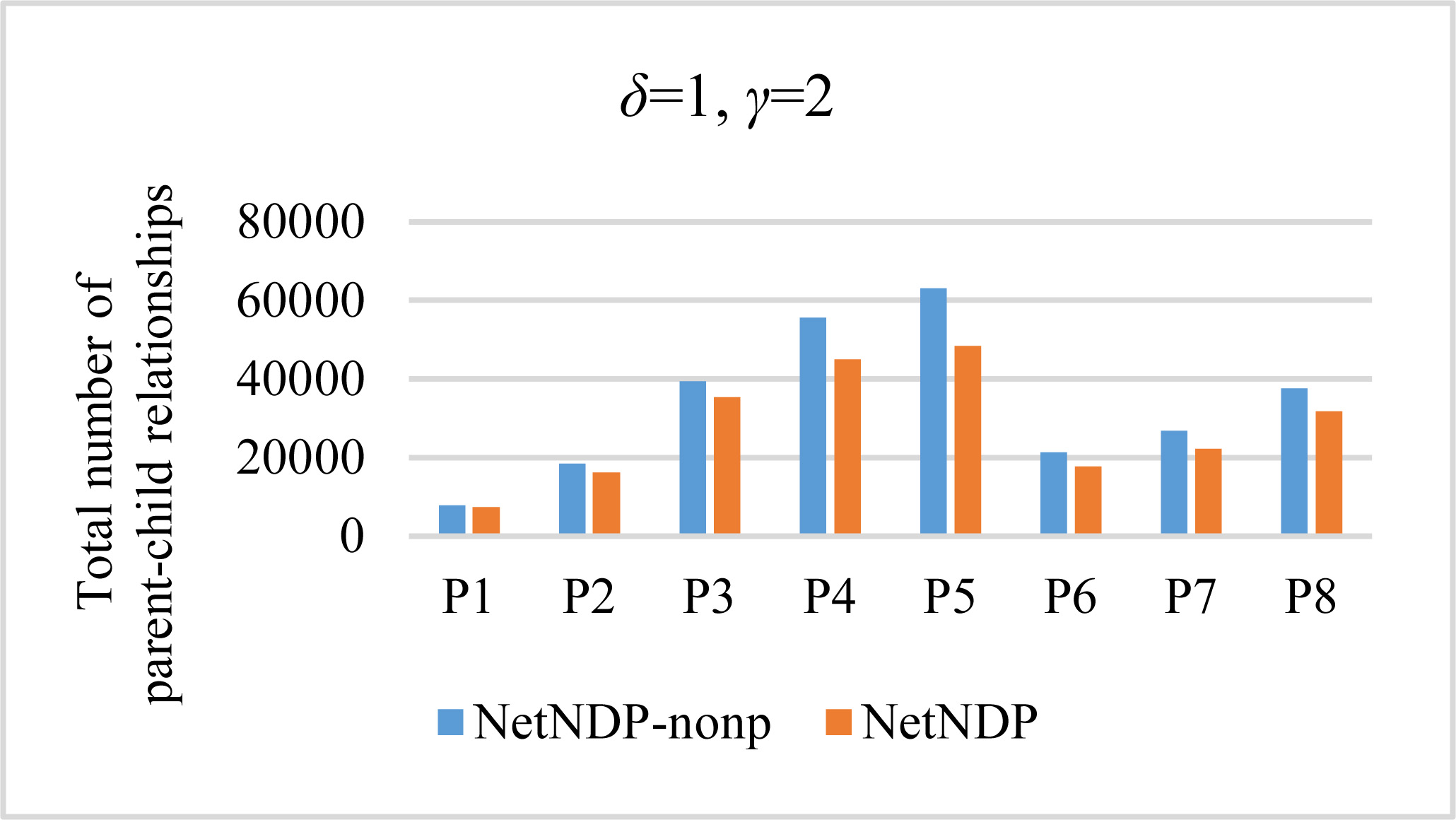}
		\end{minipage}%
		\caption{Total numbers of nodes and parent-child relationships for NetNDP-nonp and NetNDP with $(\delta=1, \gamma=2)$}
		\label{fig10}
	\end{figure}
	
From Figure \ref{fig10}, we can see that NetNDP-nonp has numerous invalid nodes and parent-child relationships, and therefore requires excessive numbers of invalid accesses, thus increasing the running time. Hence, NetNDP has better performance than NETASPNO-$(\delta, \gamma)$ and NetNDP-nonp.

	\subsection{Influence of Different Parameters}
	\label{parameters}
	To further demonstrate the impact of different threshold parameters on the experiment, we use values of $(\delta=1, \gamma=3)$ and $(\delta=2, \gamma=3)$. The corresponding numbers of occurrences are shown in Figures \ref{fig11} and \ref{fig12}, and the total running time is shown in Figures \ref{fig13} and \ref{fig14}.

		\begin{figure}[t] 
		\centering
		\includegraphics[width=0.5\textwidth]{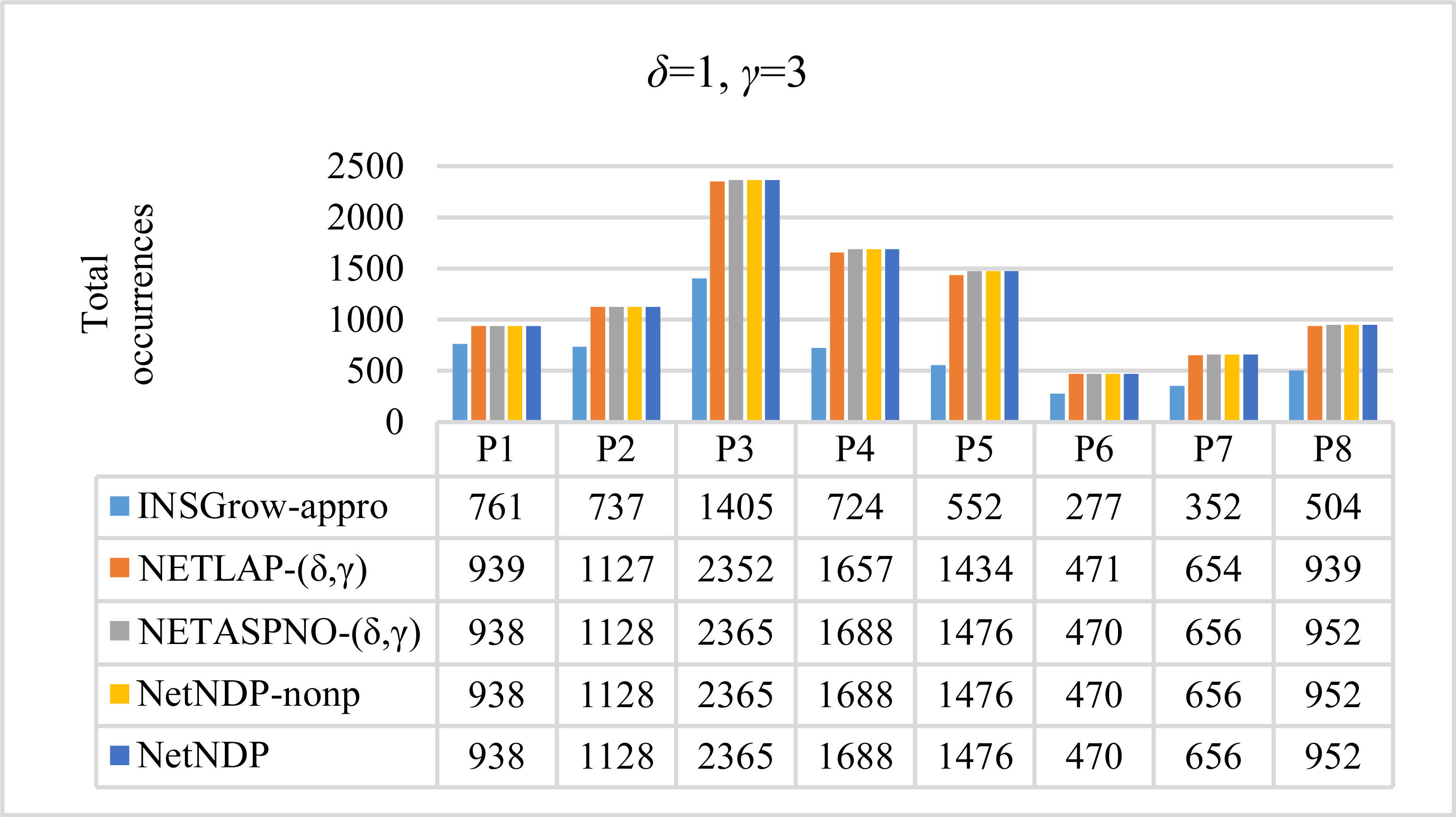}
		\caption{Total occurrences for $P1\sim P8$ on $S1\sim S8$ with $(\delta=1, \gamma=3)$}
		\label{fig11} 	
	\end{figure}
	
	\begin{figure}[t] 
		\centering
		\includegraphics[width=0.5\textwidth]{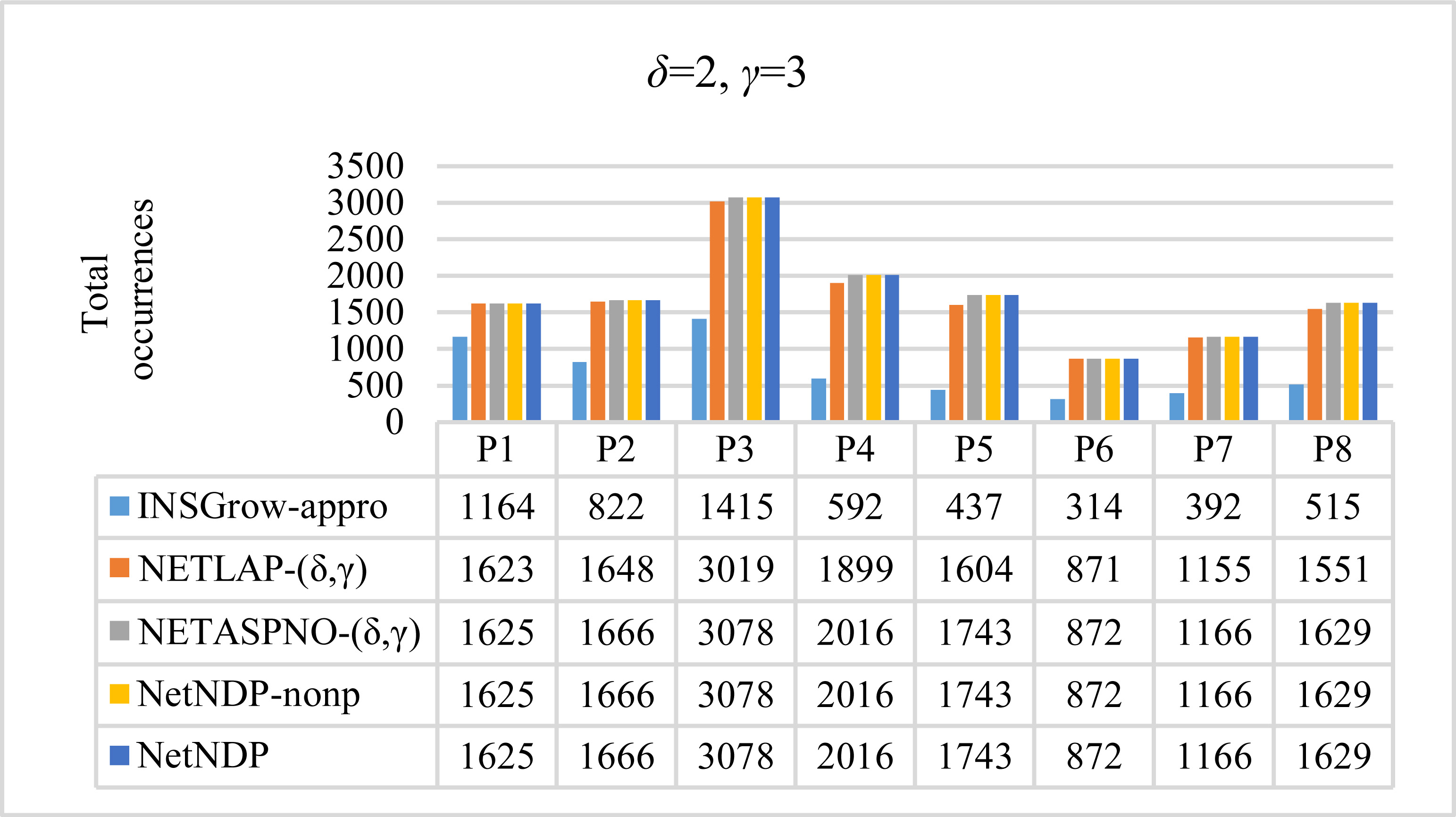}
		\caption{Total occurrences for $P1\sim P8$ on $S1\sim S8$ with $(\delta=2, \gamma=3)$}
		\label{fig12} 	
	\end{figure}
	\begin{figure}[t] 
		\centering
		\includegraphics[width=0.5\textwidth]{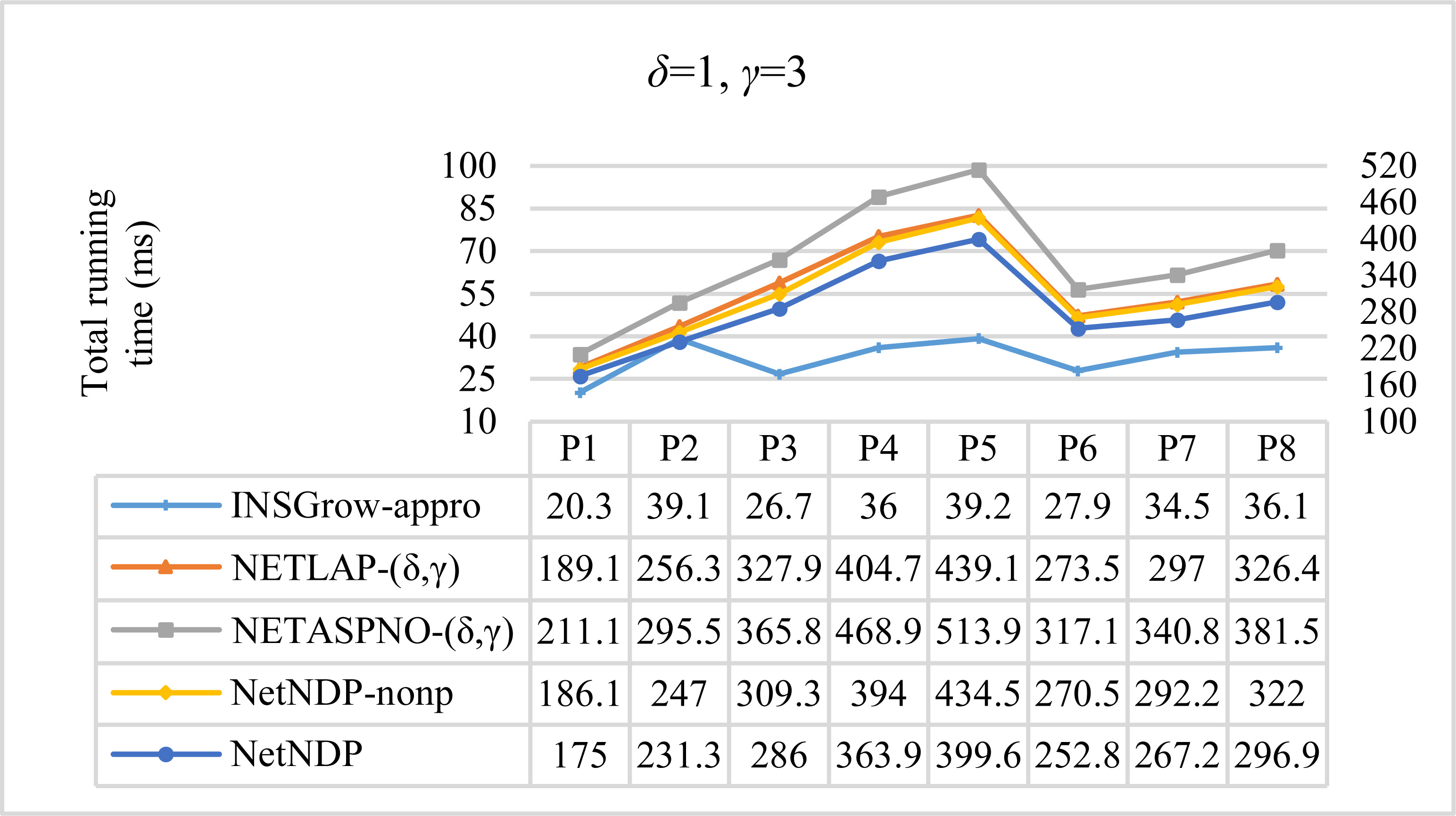}
		\caption{Total running time for $P1\sim P8$ on $S1\sim S8$ with $(\delta=1, \gamma=3)$}
		\label{fig13} 	
	\end{figure}
	
	\begin{figure}[t] 
		\centering
		\includegraphics[width=0.5\textwidth]{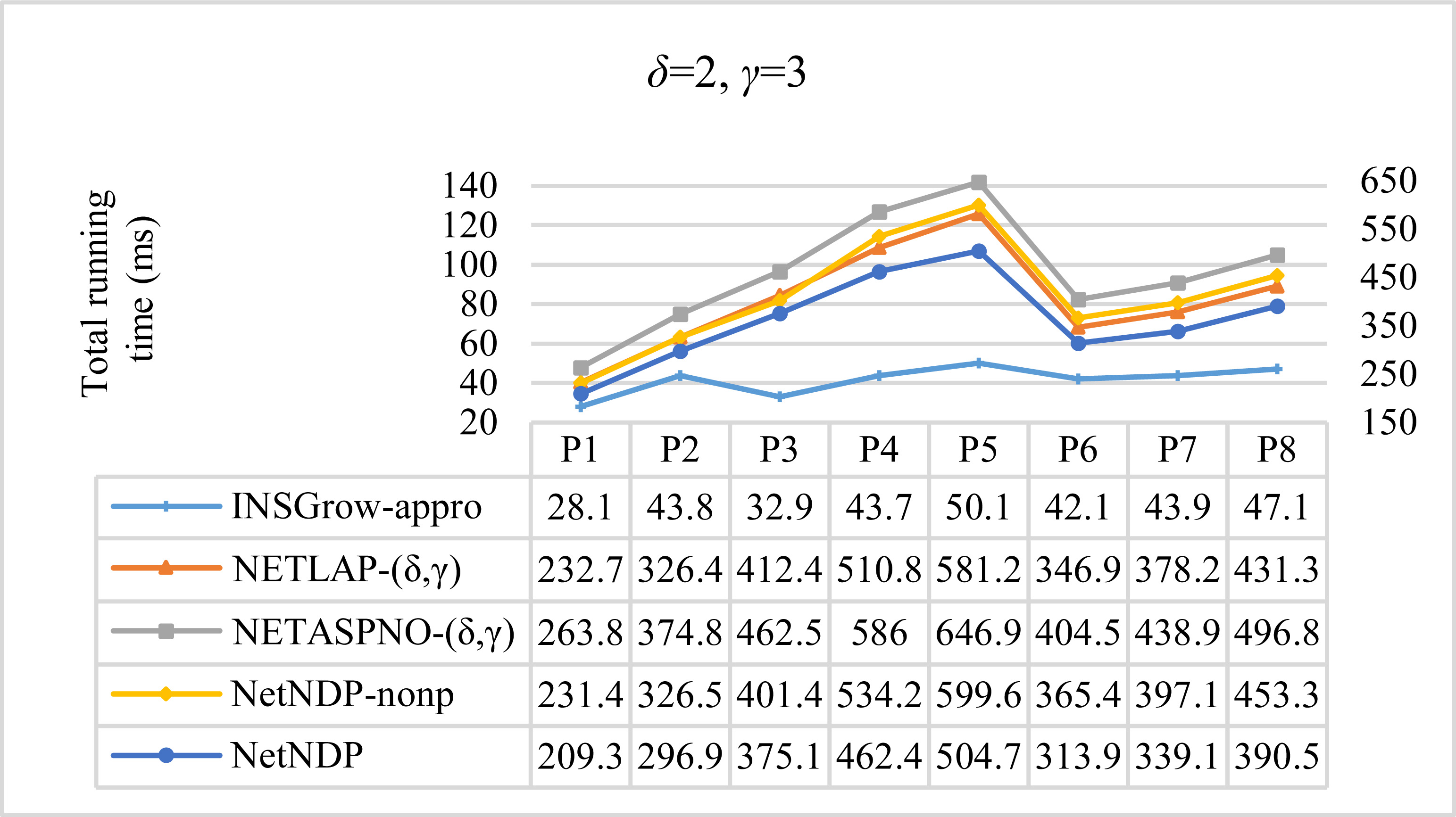}
		\caption{Total running time for $P1\sim P8$ on $S1\sim S8$ with $(\delta=2, \gamma=3)$}
		\label{fig14} 	
	\end{figure}

	(1) From Figures \ref{fig11} and \ref{fig12}, we can see that NetNDP finds the same number of occurrences as NETASPNO-$(\delta, \gamma)$ and NetNDP-nonp, and that its performance is better than INSGrow-appro and NTELAP-$(\delta, \gamma)$. Figures \ref{fig13} and \ref{fig14} show that except for INSGrow-appro, NetNDP is the fastest. The  analysis given above is consistent with the results for $(\delta=1, \gamma=2)$.
	
	(2) As can be seen from Figures \ref{fig9}, \ref{fig13}, and \ref{fig14}, the running time for NetNDP is positively correlated with $n$, $m$, and $W$, a result that is consistent with our theoretical analysis of a time complexity $O(n \times m^{2} \times W)$.

	From Table 2, we can see that the length of $S1$ is the shortest and the length of $S8$ is the longest. Figure \ref{fig9} shows that the running time for $P1\sim P8$ is the shortest on $S1$ and the longest on $S8$; in other words, the larger the value of $n$, the longer the running time. For instance, the running time for $P4$ on $S1$ and $S8$ is 25 and 65.6 ms, respectively, while the running time for $P4$ on the other sequences is greater than 25 ms and less than 65.6ms. Hence, the running time for NetNDP is positively correlated with the length of the sequence.
	
	Table 3 shows that $P3\sim P5$ have the same gap constraints and that the length increases gradually. From Figures \ref{fig13} and \ref{fig14}, we know that the total running time for $P5$ is always greater than for $P3$ on $S1\sim S8$. For example, the total running time for $P3$ and $P5$ on $S1\sim S8$ is 286 and 399.6ms $(\delta=1, \gamma=3)$, respectively. These experimental results demonstrate that the running time increases with the length of the pattern.

	From Table 3, we see that $P6\sim P8$ are same except for the gap constraints and that their maximum gap increases gradually. Figures \ref{fig13} and \ref{fig14} show that the total running time for $P8$ is always greater than that of $P6$ on $S1\sim S8$, indicating that the longer the maximum gap, the longer the running time.

	In summary, the running time of NetNDP is consistent with a time complexity of $O(n \times m^{2} \times W)$, and NetNDP outperforms other competitive algorithms.

	\subsection{Matching Effect}
	\label{matching_effect}
	To illustrate that the $(\delta, \gamma)$-distance outperforms the Hamming distance in terms of the matching effect, we obtain occurrences with the $(\delta, \gamma)$-distance using NetNDP $(\delta=1, \gamma=3)$, and occurrences with a Hamming distance using NETASPNO \cite{Wu2018NETASPNO} $(h=3)$. Figures \ref{fig15} and \ref{fig16} show the matching results for $P9$ on $S9$ and $P10$ on $S10$, respectively. 

	\begin{figure}[t] 
		\begin{minipage}[t]{0.5\linewidth}	
			\includegraphics[width=2.3in]{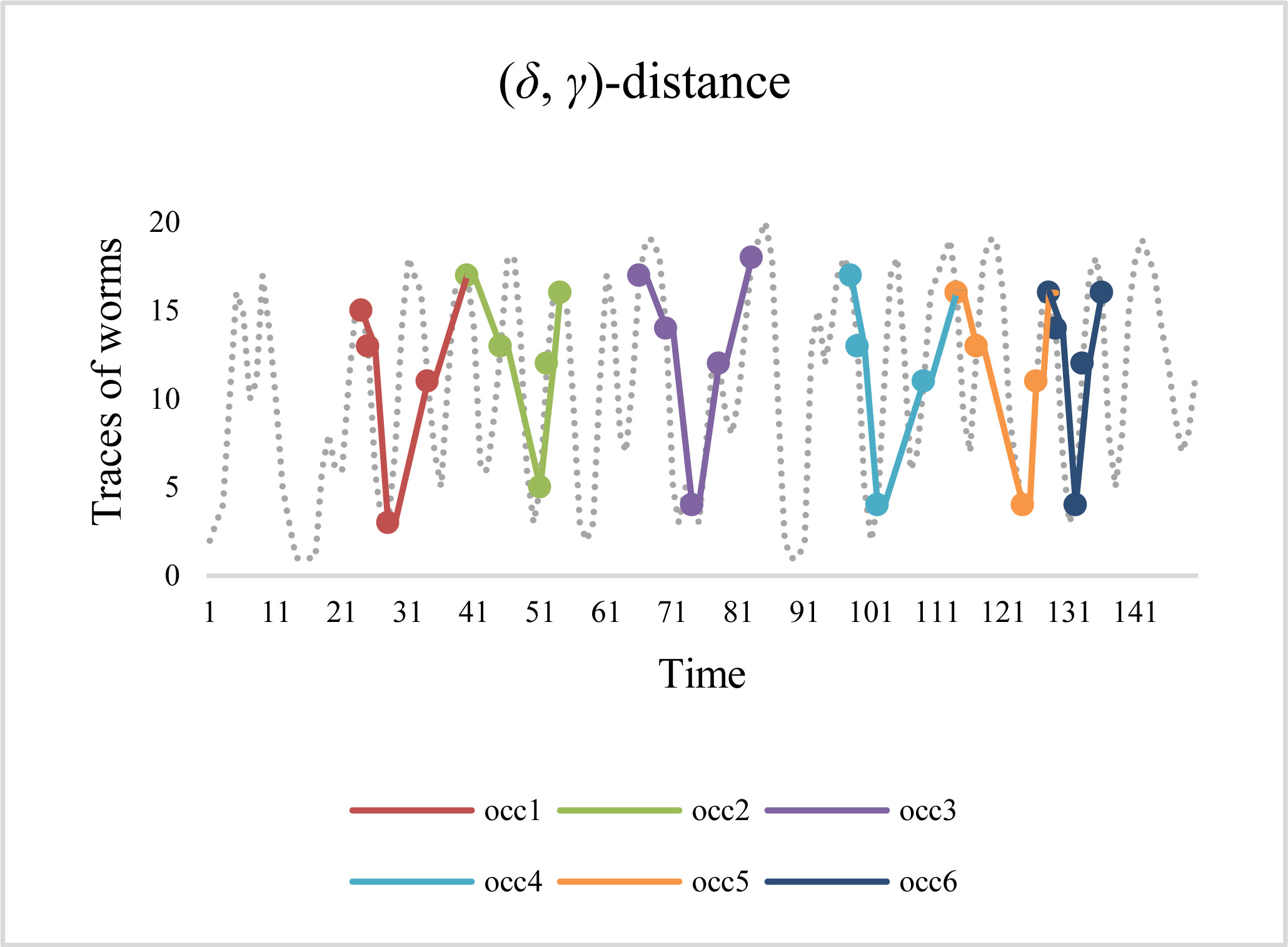}
		\end{minipage}%
		\begin{minipage}[t]{0.5\linewidth}	
			\includegraphics[width=2.3in]{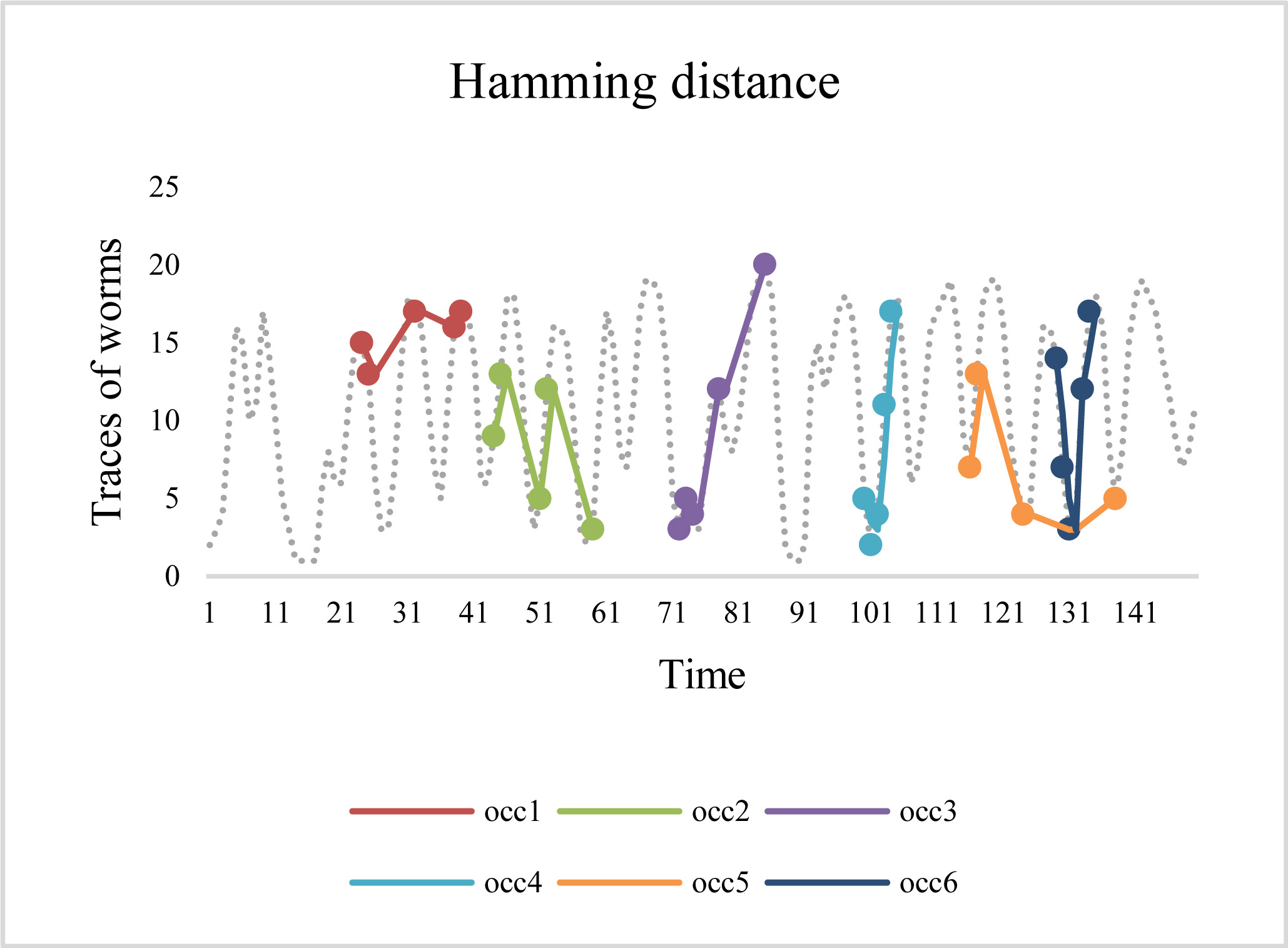}
		\end{minipage}%
		\caption{Matching results for $P9$ on $S9$ with $(\delta=1, \gamma=3, h=3)$}
		\label{fig15}
	\end{figure}

	\begin{figure} 
		\begin{minipage}[t]{0.5\linewidth}	
			\includegraphics[width=2.3in]{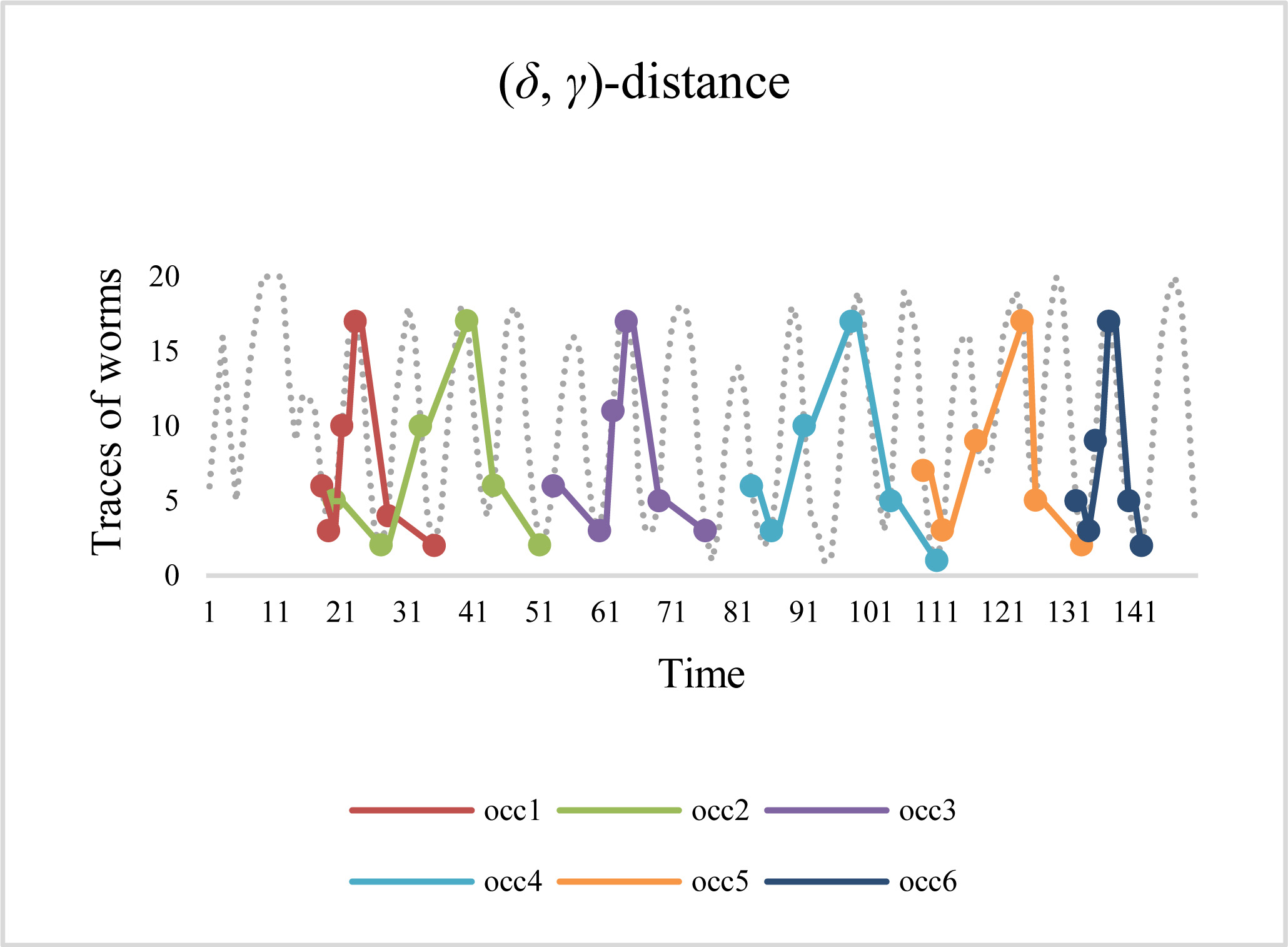}
		\end{minipage}%
		\begin{minipage}[t]{0.5\linewidth}	
			\includegraphics[width=2.3in]{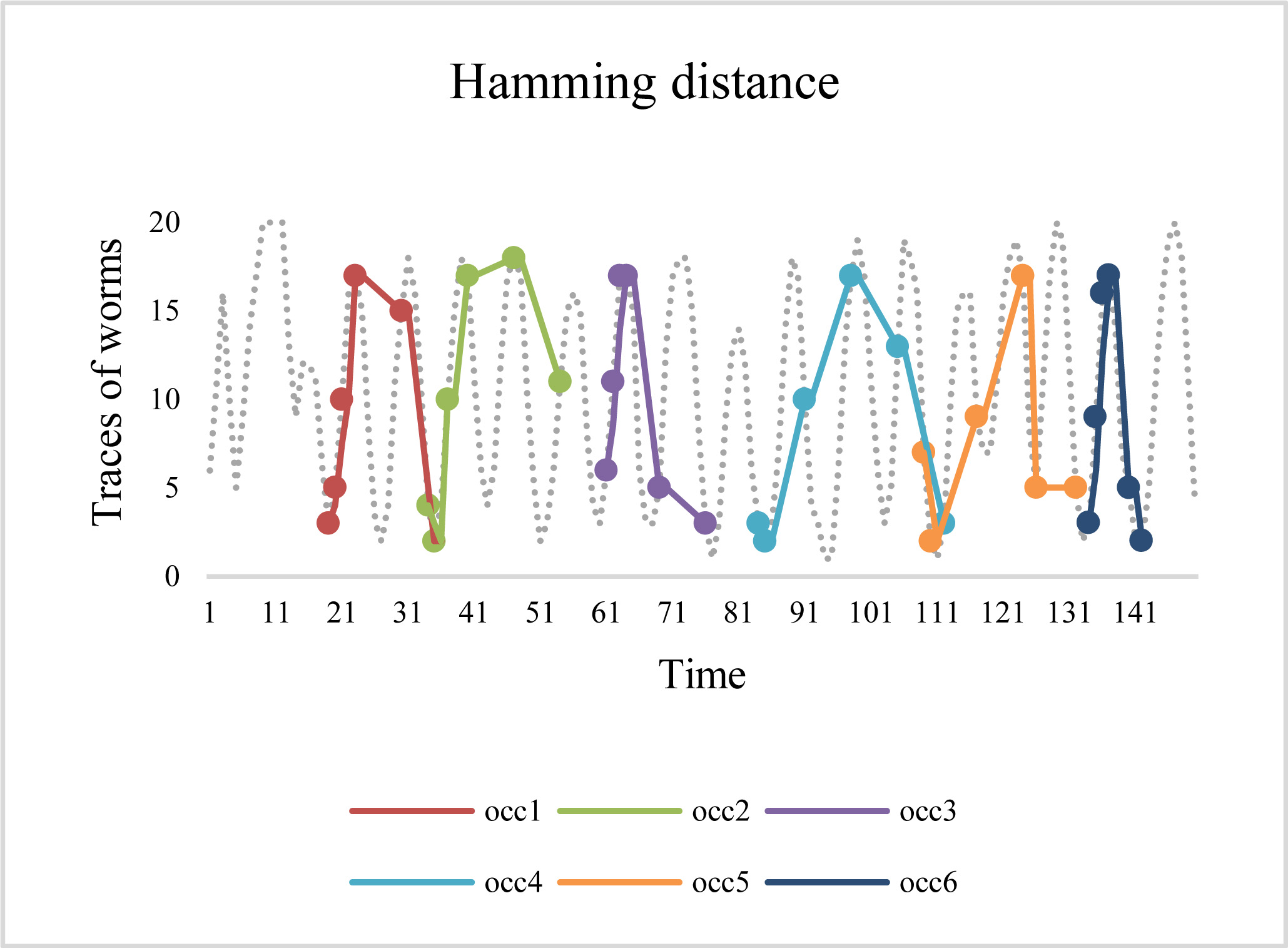}
		\end{minipage}%
		\caption{Matching results for $P10$ on $S10$ with $(\delta=1, \gamma=3, h=3)$}
		\label{fig16}
	\end{figure}
	As can be seen from Figure \ref{fig15}, the trends of all occurrences found with the $(\delta, \gamma)$-distance are the same that of $P9$, while the trends of some occurrences with the Hamming distance are different from that of $P9$. The reason for this is that the Hamming distance cannot reflect the distance between characters, i.e. it cannot measure the local approximation, resulting in large deviations in matching results. For instance, with the  Hamming distance, $occ1$ has a large deviation from $P9$ in position 32, which leads to a dissimilarity between $occ1$ and $P9$. However, when the $(\delta, \gamma)$-distance is used, the distance between the corresponding characters of occ1 and $P9$ is less than one due to the local constraint, and $occ1$ therefore is similar to $P9$. Figure \ref{fig16} shows a similar phenomenon. Based on this analysis, we know that approximate pattern matching with the $(\delta, \gamma)$-distance is more effective than with the Hamming distance.
	
	\section{Conclusion}
	\label{conclusion}
	
The Hamming distance cannot be used to measure the local approximation between the subsequence and pattern, resulting in large deviations in matching results. To overcome this weakness of the Hamming distance, we explore the use of nonpoverlapping approximate pattern matching with the $(\delta, \gamma)$-distance, where the $\delta$-distance and $\gamma$-distance are used to measure the local and the global approximations, respectively. 	We develop the concept of a local approximate Nettree, and construct an efficient algorithm called NetNDP based on a local approximate Nettree.To improve the time efficiency, NetNDP employs MRD to prune invalid nodes and parent-child relationships, and to assess whether or not the root paths satisfy the global constraint. NetNDP finds the rightmost absolute leaf of the max root, searches for the rightmost occurrence from the rightmost absolute leaf, and deletes the occurrence. These processes are iterated until there are no new occurrences in the local approximate Nettree. Numerous experimental results give rise to the following conclusions. NetNDP runs faster than other competitive algorithms, since it can avoid creating invalid nodes and parent-child relationships. More importantly, approximate pattern matching with the $(\delta, \gamma)$-distance has better matching performance than the Hamming distance, since the trends of all occurrences found with the $(\delta, \gamma)$-distance are the same with that of the pattern, while the trends of some occurrences with the Hamming distance are different from that of the pattern.

	\section*{Acknowledgement}
		This work was party supported by National Natural Science Foundation of China (61976240, 917446209), National Key Research and Development Program of China(2016YFB1000901), Natural Science Foundation of Hebei Province, China (No. F2020202013).

\end{document}